\relax
\documentclass[11pt]{article}
\usepackage{times}
\usepackage{helvet}
\usepackage{courier}
\usepackage{url}
\usepackage{graphicx}

\usepackage{color}
\usepackage{multirow}
\usepackage{amssymb}
\usepackage{amsmath,amsfonts,latexsym,graphicx}
\usepackage{cases}
\usepackage[linesnumbered,boxed,ruled,vlined]{algorithm2e}
\usepackage{nicefrac}
\usepackage{enumitem}
\usepackage{xcolor,colortbl}

\newtheorem{lemma}{Lemma}
\newtheorem{theorem}[lemma]{Theorem}
\newtheorem{defn}[lemma]{Definition}

\newtheorem{corollary}[lemma]{Corollary}

\newtheorem{claim}[lemma]{Lemma}
\newtheorem{remark}[lemma]{Remark}

\newtheorem{example}{Example}[section]

\newcommand{\score}{{\mathrm{score}}}
\newcommand{\pos}{{\mathrm{pos}}}
        {\medskip}

\newenvironment{proof}{\noindent{\bf Proof:}}%
        {\hspace*{\fill}$\Box$\par}
        {\hspace*{\fill}$\Box$\par}
        {\hspace*{\fill}$\Box$\par}

\newcommand{\topic}[1]{\noindent{{\bf #1}:}}

\newcommand{\eps}{\varepsilon}
\newcommand{\Prob}{{\operatorname{Pr}}}
\newcommand{\Exp}{{\mathbb{E}}}

\newcommand{\ceil}[1]{\big\lceil #1 \big\rceil}
\newcommand{\floor}[1]{\big\lfloor #1 \big\rfloor}
\newcommand{\eat}[1]{}
\newcommand{\e}{\mathbf{e}}
\newcommand{\R}{\mathbb{R}}
\newcommand{\Z}{\mathbb{Z}}

\newcommand{\calB}{\mathcal{B}}
\newcommand{\calI}{\mathcal{I}}

\newcommand{\calN}{\mathcal{N}}

\title{\bf Multiwinner Voting with Fairness Constraints}
\author{L. Elisa Celis, Lingxiao Huang, and Nisheeth K. Vishnoi \\
EPFL, Switzerland \\
}

\begin{document}

\maketitle

\begin{abstract}
\label{sec:abstract}
Multiwinner voting rules are used to select a small representative subset of candidates or items from a larger set given the preferences of voters.
However, if candidates have sensitive attributes such as gender or ethnicity  (when selecting a committee), or specified types such as political leaning (when selecting a subset of news items), an algorithm that chooses a subset by optimizing a multiwinner voting rule may be unbalanced in its selection -- it may under or over represent a particular gender or political orientation in the examples above.
We introduce an algorithmic framework for multiwinner voting problems when there is an additional requirement that the selected subset should be ``fair'' with respect to a given set of attributes.
Our framework provides the flexibility to (1) specify fairness with respect to {\em multiple, non-disjoint} attributes (e.g., ethnicity \emph{and} gender) and (2) specify a score function.
We study the computational complexity of this constrained multiwinner voting problem for monotone and submodular score functions and present several approximation algorithms and matching  hardness of approximation results  for various attribute group structure and types of score functions.
We also present simulations that suggest that adding fairness constraints may not affect  the scores significantly when compared to the unconstrained case.
\end{abstract}

\newpage
\tableofcontents

\newpage

\section{Introduction}
\label{sec:intro}

The problem of selecting a committee from a set of candidates given the preferences of voters, called a multiwinner voting problem, arises in various social, political and e-commerce settings -- from electing a parliament to govern a country, to selecting a committee to decide prizes and awards, to the selection of products to display on Amazon's front page or which movies and TV shows to present on Netflix.
Formally, there is a set $C$ of $m$ ``candidates'' (i.e., people, products, articles, or other items) that can be selected and a set $A$ of $n$ ``voters'' (who either vote explicitly or are simply users who implicitly state preferences via their behavior) that has a (possibly incomplete) preference list over the $m$ candidates.
The goal is to select a subset of $C$ of size $k$ based on these preferences.

 Given the preference lists, it remains to specify how the selection will be made.
One common approach is to define
 a ``total'' score function; namely, a function $\score: 2^C\rightarrow \R_{\geq 0}$ that gives a score to each potential committee which depends on the voters' preferences.
This reduces the selection to an optimization problem:  pick a  committee of size $k$  that maximizes the score.
Different views on the desired properties of the selection process have led to
a number of different scoring rules and, consequently, to a variety of algorithmic problems that have been a topic of much recent  interest; see~\cite{faliszewski2017multiwinner}.
Prevalent examples include general multiwinner voting rules such as the committee scoring rules~\cite{elkind2014properties,aziz2017condorcet}, the approval-based rules~\cite{aziz2017justified}, the OWA-based rules~\cite{skowron2016finding}, variants of the Monroe rule~\cite{betzler2013computation,monroe1995fully,skowron2015achieving}
and the goalbase rules~\cite{uckelman2009representing}.

However, it has been shown that voting rules, in the most general sense, can negatively affect the proportion of women in the US legislature~\cite{votingrules},
and result in a disproportionate electorate that under-represents a minority ~\cite{DBLP:journals/corr/FaliszewskiLSSS16}.
Furthermore, such algorithmic biases have been shown to influence human perceptions and opinions of such minorities \cite{kay2015unequal}.
An increasing awareness
has led governments to generic~\cite{press2016preparing} and specific \cite{zealand1985royal} recommendations 
in order to ensure sufficient representation of minority populations.

In response, ``proportional representation'' rules
\cite{monroe1995fully,brill2017multiwinner},
that ensure that the political affiliations of the electorate are reflected proportionately in the elected body, are being deployed.
Formally, let $P_1,\ldots,P_p \subseteq [m]$ be $p$  {\em disjoint} groups of  candidates where $i \in \left\{1, \ldots, p\right\}$ represents a given group.
For any $i\in [p]$, let $f_i$ be the fraction of voters who belong to (or, more generally, prefer) group $i$.
Then, a voting rule achieving full proportionality would ensure that the selected committee $S$ satisfies $\floor{k \cdot f_i} \leq |S \cap P_i| \leq \ceil{k \cdot f_i}$; see also~\cite[Definition 5]{brill2017multiwinner}.
{Other proportional representative schemes include the proportional approval voting (PAV) rule~\cite{brill2017multiwinner} and the Chamberlain Courant-rule~\cite{chamberlin1983representative}.}
Other approaches, such as Brill et al.~\cite{brill2017multiwinner}, consider different notions of ``fair'' representation; e.g., Koriyama et al.~\cite{koriyama2013optimal} argue that minorities should have disproportionately many representatives.

The main conceptual contribution of this paper is the first algorithmic framework for multiwinner voting problems that can (1) incorporate very general notions of fairness with respect to arbitrary group structures (which include multiple simultaneous attributes) and (2) outputs a subset that maximizes the given score function subject to being fair.
Our model and the algorithmic problem arising therein is described
in Section \ref{sec:model}.
The main technical contributions
of the paper are a host of new approximate algorithms and hardness of approximation results depending on the score function and the precise structure of the fairness constraints.
The hardness results are  described
in Section \ref{sec:hard} and the algorithms in Section \ref{sec:alg}.
Empirically, we show that existing multiwinner voting rules may introduce bias and that our approach not only ensures fairness, but does so with a score that is close to the (unconstrained) optimal (see Section \ref{sec:experiment}).
Overall, our work gives a promising
general and rigorous algorithmic solution to the problem of
controlling and alleviating bias arising from various current multiwinner voting settings.

\section{Our Contributions}\label{sec:model}

\paragraph{Model.}
In a multiwinner voting setting, we are given $m$ candidates, $n$ voters that each has a
(potentially incomplete and/or non-strict) preference list
over the $m$ candidates,
a {score function} $\score:2^{[m]} \to \mathbb{R}_{\geq 0}$ defined by these lists,
and a desired number $k\in [m]$ of winners.
In addition, to consider fair solutions, we are given arbitrary (potentially non-disjoint) groups of candidates $P_1,\ldots,P_p \subseteq C$,
and fairness constraints on the selected winner set $S$ of the form: $\ell_i\leq |S\cap P_i|\leq u_i, ~ \forall i\in [p]$
for given numbers $\ell_1,\ldots,\ell_p,u_1,\ldots,u_p\in \Z_{\geq 0}$.
The goal of the constrained multiwinner voting problem is to select a committee of size $k$ that maximizes $\score(S)$ and satisfies all fairness constraints. If the score function is monotone and submodular we call the problem the {\em constrained MS multiwinner voting problem}. 
This includes the well-studied Chamberlin-Courant (CC) rule, the Monroe rule, the OWA-based rules and the goalbase rules.

\begin{table*}[t]
  \centering
{\tiny
  \begin{tabular}{|*{7}{c|}}
\hline
 &  $\Delta\geq 3$ & $\Delta=1$ & only $u_i$ & only $\ell_i$ & $p=O(1)$ & Unconstrained \\ \cline{1-7}
  \cellcolor[gray]{0.8} A & \cellcolor[gray]{0.8} $(1-\nicefrac{1}{e}-o(1))$-bi (Thm. \ref{thm:general})   & \cellcolor[gray]{0.8} $1-\nicefrac{1}{e} $ (Thm. \ref{thm:delta=1})  & \cellcolor[gray]{0.8} $\frac{1}{\Delta+1}$-bi (Thm. \ref{thm:onlyupper})  & \cellcolor[gray]{0.8} $1-\nicefrac{1}{e}-o(1)$  (Thm. \ref{thm:onlylower}) & \cellcolor[gray]{0.8} $1-\nicefrac{1}{e}$  (Thm. \ref{thm:constant})  & \cellcolor[gray]{0.8} $1-\nicefrac{1}{e}$  (NWF78a)\\  \cline{1-7}
   C & Feasibility NP-hard (Thm. \ref{thm:Delta3hard}-\ref{thm:constantviolation}) & $1-\nicefrac{1}{e}+\varepsilon$  &  $O(\nicefrac{\log \Delta}{\Delta}) $ (Thm. \ref{thm:approxhard}) & Feasibility NP-hard  (Thm. \ref{thm:lowerhard})  & $1-\nicefrac{1}{e}+\varepsilon$   & $1-\nicefrac{1}{e}+\varepsilon$ (NW78) \\  \hline
  \end{tabular}
  \caption{\footnotesize A summary of the results for the constrained monotone submodular (MS) multiwinner voting problem; each column denotes a different kind of group structure of the attributes.
  In Row ``A'' (for algorithm),  the entry ``$\theta$'' means that there exists a $\theta$-approximation algorithm;
 ``$\theta$-bi'' means that the algorithm produces a $\theta$-approximate solution compared to the optimal solution but the fairness constraints may be violated by a small multiplicative factor.
   In Row ``C'' (for complexity), the entry ``$\theta$'' means that under the assumption $P\neq NP$, there does not exist a polynomial time algorithm that will always find a solution whose value is at most a $\theta$ factor from the optimal solution;
  ``Feasibility NP-hard'' means that it is NP-hard to check whether there is a feasible solution satisfying all fairness constraints.
  $\varepsilon>0$ is an arbitrary constant.
For each entry, we point the reader to either the reference for the result, or our Theorem $i$ (Thm. $i$) which contains it.
}
  \label{tab:msresult}}
\end{table*} \noindent

\begin{table*}[t]
  \centering
\tiny{
  \begin{tabular}{|*{7}{c|}}
\hline
 \multicolumn{2}{|c|}{Voting Rules} &
 $\Delta\geq 3$  & $\Delta=2$ & $\Delta=1$ & $p=O(1)$ &  Unconstrained \\
 \hline \cline{1-7}
 \multirow{2}{*}{SNTV} &
 \cellcolor[gray]{0.8} A  &  \cellcolor[gray]{0.8} $(1-o(1))$-bi  (Thm. \ref{thm:general})   & \cellcolor[gray]{0.8} P (Thm. \ref{thm:SNTV2}) & \cellcolor[gray]{0.8} P  (Thm. \ref{thm:SNTVdelta=1})  & \cellcolor[gray]{0.8} P  (Thm. \ref{thm:constant}) & \cellcolor[gray]{0.8} P \\  \cline{2-7}
  & C  &Feasibility NP-hard  (Thm. \ref{thm:Delta3hard}-\ref{thm:constantviolation})  & P & P & P & P \\ \cline{1-7}
 \multirow{2}{*}{$\alpha$-CC} &
 \cellcolor[gray]{0.8}  A   &  \cellcolor[gray]{0.8} $(1-\nicefrac{1}{e}-o(1))$-bi  (Thm. \ref{thm:general})   &  \cellcolor[gray]{0.8} $(1-\nicefrac{1}{e}-o(1))$-bi  (Thm. \ref{thm:general})   &  \cellcolor[gray]{0.8} $1-\nicefrac{1}{e} $  (Thm. \ref{thm:delta=1}) &  \cellcolor[gray]{0.8}$1-\nicefrac{1}{e}$  (Thm. \ref{thm:constant})  &  \cellcolor[gray]{0.8} $1-\nicefrac{1}{e}$  (LB11) \\  \cline{2-7}
  & C  & Feasibility NP-hard  (Thm. \ref{thm:Delta3hard}-\ref{thm:constantviolation}) &$1-\nicefrac{1}{e}+\varepsilon$  (SFS15) &$1-\nicefrac{1}{e}+\varepsilon$  (SFS15) & $1-\nicefrac{1}{e}+\varepsilon$  (SFS15) & $1-\nicefrac{1}{e}+\varepsilon$  (SFS15) \\  \cline{1-7}
 \multirow{2}{*}{$\beta$-CC} &
 \cellcolor[gray]{0.8} A   &  \cellcolor[gray]{0.8} $(1-\nicefrac{1}{e}-o(1))$-bi  (Thm. \ref{thm:general})   &  \cellcolor[gray]{0.8} $(1-\nicefrac{1}{e}-o(1))$-bi  (Thm. \ref{thm:general})   & \cellcolor[gray]{0.8} $1-\nicefrac{1}{e} $  (Thm. \ref{thm:delta=1}) & \cellcolor[gray]{0.8} $1-\nicefrac{1}{e}$  (Thm. \ref{thm:constant})  & \cellcolor[gray]{0.8} PTAS (SFS15) \\  \cline{2-7}
  & C  & Feasibility NP-hard  (Thm. \ref{thm:Delta3hard}-\ref{thm:constantviolation})  & $1-\nicefrac{1}{e}+\varepsilon$  (Thm. \ref{thm:simplelower})  &$1-\nicefrac{1}{e}+\varepsilon$  (Thm. \ref{thm:simplelower})  &$1-\nicefrac{1}{e}+\varepsilon$  (Thm. \ref{thm:simplelower}) & NP-hard  (PRZ08) \\\hline
  \end{tabular}
  \caption{\footnotesize
  A summary of our results for the constrained monotone submodular (MS) multiwinner voting problem using three variants of the Chamberlin-Courant rule; each column denotes a different kind of group structure of the attributes.
  The definitions of (Thm. $i$), ``A'', ``C'', ``$\theta$-bi'', ``Feasibility NP-hard'' and $\varepsilon$ are the same as in Table~\ref{tab:msresult}. ``P'' means there exists a polynomial time exact algorithm.
     }
  \label{tab:experiment}}
 \end{table*}

\paragraph{Results.} As our model generalizes prior work on the unconstrained multiwinner voting case, the algorithmic problems that arise largely remain NP-hard and we focus on developing approximation algorithms for them.
An important practical parameter, that also plays a role in the complexity of the constrained multiwinner voting problem, is the maximum number of groups in which any candidate can be; we denote it by $\Delta.$
In real-world situations, we expect $\Delta$ to be a small constant, i.e.,  each committee member is in only a few groups.
\eat{
\textcolor{red}{For instance, consider a newspaper that wants to display some summarizations of news on the first page.
However, the space is limited and can only include $k$ summarizations.
How should the newspaper select summarizations to display?
In practice, the newspaper may require that each major topic, e.g., politics, business, sports, and entertainment, should have 1-3 summarizations.
This requirement corresponds to the case of $\Delta=1$.
Sometimes, the newspaper also requires that different views of a controversial topic should have at least 1 summarization.
Including this requirement makes the parameter $\Delta=2$.
Moreover, the newspaper may require 1-2 summarizations of the in-depth report on the first page or have other requirements on contents, which makes $\Delta=3$.
}
}
Our main results (classified by $\Delta$) are summarized  in Tables~\ref{tab:msresult} and~\ref{tab:experiment}.

When $\Delta=1$, that is when each committee member can be a part of at most one group, we present a $(1-1/e)$-approximation algorithm (Theorem~\ref{thm:delta=1}) which matches the $(1-1/e-\varepsilon)$-hardness of approximation \cite{nemhauser1978analysis}.

When $\Delta \geq 3,$  unlike the unconstrained case, even checking whether there is a feasible solution becomes NP-hard (Theorem~\ref{thm:Delta3hard}).
Further, the problem of finding a solution that violates cardinality or fairness constraints up to any multiplicative constant factor remains NP-hard (Theorems~\ref{thm:lowerhard} and~\ref{thm:constantviolation}).
Moreover, even if feasibility is guaranteed, the problem remains hard to approximate within a factor of $\Omega(\log \Delta/\Delta)$ (Theorem~\ref{thm:approxhard}).
To bypass the issue of  ensuring feasibility, we assume the problem instance always has a feasible solution and present certain sufficient conditions on the input instances that {\em guarantee} feasibility.
For instance, if we assume that the fraction of each group in the selected committee is allowed some slack when compared to their proportion in the set of candidates  ($\ell_i \leq k (|P_i|/m-0.05)$ and $u_i \geq k (|P_i|/m+0.05)$).
then a random committee of size $k$ is likely to be feasible.
   We discuss this and other natural assumptions that can ensure feasibility in Section~\ref{sec:feasibilityconditions}.
Subsequently, we give a near-optimal bi-criterion approximation algorithm violating each fairness constraint by a small multiplicative factor for the general class of MS voting rules (Theorem~\ref{thm:general}).

We also study special cases: the fairness constraints involve only lower bounds (Theorem~\ref{thm:onlylower}), only upper bounds (Theorem~\ref{thm:onlyupper}), or there is a constant number of fairness constraints (Theorem~\ref{thm:constant}).

Finally, we also study some special voting rules with MS score functions such as SNTV, $\alpha$-CC and $\beta$-CC (see Section~\ref{sec:pre}
for definitions) where the unconstrained problem has recently received considerable attention.
See Table~\ref{tab:experiment} for a summary of our results for these rules.
Unlike the unconstrained case where a PTAS is known, we show that constrained $\beta$-CC multiwinner voting is $(1-1/e+\varepsilon)$-inapproximable ($\varepsilon>0$ is constant) even if $\Delta=1$ and $p=2$ (Theorem~\ref{thm:simplelower}).
The case of $\Delta=2$ is intriguing and not entirely settled.
We prove that the constrained SNTV multiwinner voting problem has a polynomial-time exact algorithm if $\Delta=2$ (Theorem~\ref{thm:SNTV2}).

\paragraph{Techniques.}
The algorithmic results  combine  two existing tools that have been extensively used in the monotone submodular maximization literature.
The first, ``multilinear extension'' (Definition~\ref{def:multilinear}), extends the discrete MS score function to a continuous function over a relaxed domain.
By applying a continuous greedy process via multilinear extension, a fractional solution with a high score can be computed efficiently.
The second is to round the fractional solution to a  committee of size $k$  by ``dependent rounding".
In the  case of $\Delta\geq 2$ (Theorem~\ref{thm:general}), we use a swap randomized rounding procedure introduced by~\cite{chekuri2010dependent}.
In the case of $\Delta=1$ (Theorem~\ref{thm:delta=1}) we design a two-layered dependent rounding procedure that runs in  linear time.
Some of the algorithmic results are achieved by reduction to well-studied problems, like the monotone submodular maximization problem with $\Delta$-extendible system (Theorem~\ref{thm:onlyupper}) and constrained set multi-cover  (Theorem~\ref{thm:onlylower}).
The hardness results follow from reductions from well-known NP-hard problems, including $\Delta$-hypergraph matching, 3-regular vertex cover, constrained set multi-cover and independent set, and borrow techniques from a recent work on fairness for ranking problems \cite{celis2018ranking}.

\paragraph{Generality.}

This approach is general in that 1) it can handle arbitrary MS score functions, 2) multiple sensitive attributes which can take on arbitrary group structures, 3) interval constraints that need not specify exact probabilities of representations for each group, and in doing so 4) can satisfy many different existing notions of fairness.

Note that fairness can be simultaneously ensured across multiple sensitive attributes
(e.g., ethnicity \emph{and} political party) -- the number of attributes a single candidate can have is captured by $\Delta$. 
For example, in the New Zealand parliamentary election~\cite{zealand1985royal}  the parliament is required to include sufficient representation across 3 types of attributes: political parties, special interest groups, and Maori representation. 
Each candidate has an identity under each attribute, i.e., $\Delta=3$.

This type of constraints  generalize notions of proportionality that have arisen in the voting literature such as fully proportional representation \cite{monroe1995fully},
by letting $\ell_i=\big\lfloor \nicefrac{k\cdot n_i}{n}\big\rfloor$ and $u_i= \big\lceil \nicefrac{k\cdot n_i}{n}\big\rceil$ for all $i$.
(Here, $n_i$ is the number of voters who prefer type $i$.)
Similarly, one can ensure other notions such as degressive proportionality \cite{koriyama2013optimal} (e.g., satisfying Penrose's square root law~\cite{penrose1946elementary}) and flexible proportionality \cite{brill2017multiwinner}.
In particular, the percent of representation need not be exactly specified, rather one can input an allowable range.

This is also general enough to ensure the outcome satisfies existing notions of fairness, such as disparate impact, statistical parity, and risk difference.
For example, consider the case of groups that form a partition, and let  $m_i$ denote the number of voters that have type $i$.
Given some
$\xi_i\in [0,1]$, for each  $P_i$, we say committee of size $k$ satisfies $\xi$-statistical parity if $\left| \nicefrac{|S\cap P_i|}{k}-\nicefrac{m_i}{n}\right| \leq \xi_i$ (see \cite{dwork2012fairness} for the original definition).
We can set fairness constraints that guarantee $\xi$-statistical parity by setting $\ell_i$ and $u_i$ such that
$
\xi_i\geq 1-\max \left\{\left | \nicefrac{\ell_i}{k}-\nicefrac{m_i}{n}\right | ,\left | \nicefrac{u_i}{k}-\nicefrac{m_i}{n} \right | \right\}
$
for all $i$.
The groups $P_i$ and corresponding fairness parameters $\ell_i,u_i$ are taken as input and can be set according to the underlying context and desired outcome to encode a given metric;
see Section~\ref{sec:dis} for some examples.

Finally, we note that the fair multiwinner voting rule that results after adding fairness constraints continues to satisfy many nice properties (e.g., consistency, monotonicity, and fair variants of weak unanimity or committee monotonicity; see ~\cite{elkind2014properties} for formal definitions)
of the (unconstrained) voting rule; see Section~\ref{sec:dis} for details.

\section{Preliminaries}
\label{sec:pre}

Now we present the formal definition of our model and
the definitions of three monotone submodular voting rules:
SNTV, $\alpha$-CC and $\beta$-CC, which we consider as special cases.

\begin{defn}
	(Our model: constrained multiwinner voting)
	\label{def:model}
	We are given a set $C$ of $m$ candidates, $n$ voters together with
	$R = \left\{\succ_i\right\}_{i\in [n]}$ where $\succ_i$ is the preference list (potentially incomplete
	and/or non-strict) over the $m$ candidates for voter
	$i$, and a score function $\score_R : 2^{[m]} \rightarrow R_{\geq 0}$ with an evaluation
	oracle, a desired number $k\in [m]$ of winners, arbitrary
	groups $P_1,\ldots, P_p\subseteq C$ and integers $\ell_1,\ldots \ell_p,u_1,\ldots u_p\in \Z_{\geq 0}$. 
	Given a committee $S$ of size $k$, define the fairness constraints
	by $\ell_i\leq |S\cap P_i|\leq u_i, ~ \forall i\in [p]$.

	Let $\calB \subseteq 2^{[m]}$ denote the family of all committees
	of size $k$ that satisfy all fairness constraints. 
	The goal of the constrained
	multiwinner voting problem is to select an $S\in \calB$
	that maximizes $\score_R(S)$. 
	If the score function is monotone
	submodular, 
	\footnote{Recall that a function $f : 2^{[m]} \rightarrow R_{\geq 0}$ is a monotone submodular
		(MS) function if $f(A \cup B) + f(A \cap B) \leq f(A) + f(B)$ for all
		$A,B \subseteq [m]$ and $f(A) \leq f(B)$ for all $A \subseteq B$.
	}
	we call the problem the constrained MS multiwinner
	voting problem.
\end{defn}

\noindent
If $R$ is clear from the context, we denote $\score_R$ by $\score$.
This succinct description of $\calB$ and $\score$ allows us to design
fast algorithms despite the fact that $\calB$ can be exponentially.

For a preference order $\succ$ and a candidate $c \in C$, we write
$\pos_{\succ}(c)$ to denote the position of $c$ in $\succ$ (candidate ranked
first has position 1 and ranked last has position $m$). 
Given a size-$k$ committee $S\subseteq C$, we denote $\pos_{\succ}(S) := (i_1,\ldots, i_k)$ to be the sequence of positions
of the candidates in $S$ sorted in increasing order with
respect to $\succ$. Define $[m]_k$ to be the set of all size-$k$ increasing
sequences of elements from $[m]$.

\begin{defn}
	\label{def:CC}
	(CC)
	In the Chamberlin-Courant rule, there exists a positional score function $\gamma_m: [m]\rightarrow \R$ satisfying that $\gamma_m(i)\geq \gamma_m(j)$ if $1\leq i<j\leq m$. Define $\gamma_{m,k}(i_1,\ldots,i_k)=\max_{j\in [k]}\gamma_m(i_j)=\gamma_m(i_1)$ for any $(i_1,\ldots,i_k)\in [m]_k$. The total score function is defined by $\score(S):=\sum_{i\in [n]} \gamma_{m,k}\left(pos_{\succ_i}(S)\right)$.
\end{defn}

\noindent
Now we define three CC rules: SNTV, $\alpha$-CC and $\beta$-CC.

\begin{defn}
	\label{def:SNTV}
	The SNTV rule uses the following positional
	score function: 
	$\gamma_m(1) = 1$ and 
	$\gamma_m(i) = 0$ for $i > 1$. Observe that SNTV only requires that each $\succ_i$ includes the most preferred candidate of voter $i$.
\end{defn}

\begin{defn}
	\label{def:alphaCC}
	The $\alpha$-CC rule uses the following positional
	score function: 
	$\gamma_m(i) = 1$ for $i \leq k$ and 
	$\gamma_m(i) = 0$ for
	$i > k$. Observe that $\alpha$-CC only requires that each $\succ_i$ includes the
	top $k$ preferred candidates (without ordering) of voter $i$.
\end{defn}

\begin{defn}
	\label{def:betaCC}
	The $\beta$-CC rule uses the following positional score function:
	$\gamma_{m}(i)=m-i$.
\end{defn}

\noindent
If using the SNTV/$\alpha$-CC/$\beta$-CC rule, we call our framework the constrained SNTV/$\alpha$-CC/$\beta$-CC multiwinner voting problem respectively.

\begin{remark}
All three special CC rules have practical applications.
Suppose an airline wants to select some movies to display. One of the best strategies is to present a set of options which are as diverse as possible, i.e., each passenger should see
some movie appealing to him or her.
If each passenger is only satisfies by his or her favorite movie, then the SNTV rule is a natural choice. On the other hand, if each passenger has a set of good movies and is satisfied by any one of them, then it is natural to use the $\alpha$-CC rule. Finally, if each passenger has a ranking of the movies and the individual score of movies decreases linearly according to the ranking, then $\beta$-CC is our rule of choice.
\end{remark}

\section{Discussion}
\label{sec:dis}

\subsection{Generality of Fairness Constraints}
\label{sec:generality}
Our framework encompasses many existing notions of fairness that have arisen in the  machine learning literature. We summarize them in the following.

\paragraph{Fairness in the ML Literature.}

\begin{enumerate}
\item \textbf{Statistical parity \cite{dwork2012fairness}:}
Consider the case where the voters can also be partitioned into the same $p$ groups (e.g., if the groups encode ethnicity). Let  $m_i$ denote the number of voters that have type $i$.
A committee  $S$ of size $k$ has statistical parity if
 $
  \left | \frac{|S\cap P_i|}{k}-\frac{m_i}{n}\right | \approx 0
$
for all $i$.
Given some
$\xi_i\in [0,1]$ for each  $P_i$, we say committee satisfies $\xi$-statistical parity if $\left| \frac{|S\cap P_i|}{k}-\frac{m_i}{n}\right| \leq \xi_i.$
We can set fairness constraints that guarantee $\xi$-statistical parity by setting $\ell_i$ and $u_i$ such that
$
\xi_i\geq \max \left\{\left | \frac{\ell_i}{k}-\frac{m_i}{n}\right | ,\left | \frac{u_i}{k}-\frac{m_i}{n} \right | \right\}
$
for all $i$.

\item \textbf{Diversity:}
Diversity rules \cite{DBLP:conf/fie/CohoonCRL13,reyland2017how}
typically look at the population of applicants (i.e., candidates) and assert that the number of  minorities in the committee should be proportional to the number of minorities in the applicant pool, e.g.,
the 80\% rule \cite{DBLP:conf/fie/CohoonCRL13,reyland2017how}.
Thus, the goal would be to ensure that, for a committee $S$ of size $k$ satisfies
$
\left | \frac{|S\cap P_i|}{k}-\frac{|P_i|}{m}\right | \approx 0.
$
We say a committee satisfies $\xi$-diversity if $\left | \frac{|S\cap P_i|}{k}-\frac{|P_i|}{m}\right | \leq \xi_i$.
We can ensure this by setting $\ell_i$ and $u_i$ so that
$\xi_i \geq \max \left\{\left | \frac{\ell_i}{k}-\frac{|P_i|}{m}\right | ,\left | \frac{u_i}{k}-\frac{|P_i|}{m}\right | \right\}
$
for all $i$.
\end{enumerate}

\paragraph{Fairness in the Voting Literature.}
Additionally, the constraints also generalize notions of proportionality that have arisen in the voting literature such as fully proportional representation \cite{monroe1995fully}, fixed degressive proportionality \cite{koriyama2013optimal} and flexible proportionality \cite{brill2017multiwinner}; see the following. For any $i\in [p]$, let $n_i$ is the number of voters who prefer candidates who belong to type $i$.

\begin{enumerate}[resume]
\item \textbf{Unconstrained multiwinner voting:} Let $\ell_i=0$ and $u_i\geq |P_i|$ for all $i$.
\item \textbf{Fully proportional representation \cite{monroe1995fully}:}  Let $\ell_i=\big\lfloor \frac{k\cdot n_i}{n}\big\rfloor$ and $u_i= \big\lceil \frac{k\cdot n_i}{n}\big\rceil$ for all $i$.
\item \textbf{Fixed degressive proportionality \cite{koriyama2013optimal}:} For example, to ensure the  committee  satisfies the Penrose square root law, let  \cite{penrose1946elementary} $\ell_i=\big\lfloor \frac{k\sqrt{n_i}}{\sum_{i}\sqrt{n_i}}\big\rfloor$ and $u_i\geq |P_i|$ for all $i$. More generally, any fixed degressive proportionality can be attained by setting fairness parameters appropriately.
\item \textbf{Flexible proportionality \cite{brill2017multiwinner}:}
\begin{enumerate}
\item In \cite[Definition 5]{brill2017multiwinner}, the authors consider a type of flexible proportionality called \emph{lower quota}, where each group $P_i$ gets at least $\big\lfloor \frac{k\cdot n_i}{n}\big\rfloor$ seats. Then we can set $\ell_i=\big\lfloor \frac{k\cdot n_i}{n}\big\rfloor$ and $u_i= |P_i|$ for all $i$.
\item Another example of flexible proportionality is as follows. If we would be satisfied with a committee that has proportions which lie anywhere between the fully proportional and Penrose square root law, then we can set $\ell_i=\min\left\{\big\lfloor \frac{k\cdot n_i}{n}\big\rfloor,\big\lfloor \frac{k\sqrt{n_i}}{\sum_{i}\sqrt{n_i}}\big\rfloor\right\}$ and $u_i =  |P_i|$ for all $i$. Such flexibility is valuable because it can allow for a higher score; see Example~\ref{ex:flex} in Section~\ref{sec:example}.
\end{enumerate}
\end{enumerate}

\paragraph{Fairness across non-disjoint groups.}
Importantly, our fairness constraints also allow us to ensure fairness across multiple sensitive attributes -- the above example from the literature only consider groups $\{P_i\}$ that \emph{partition} the set of candidates (e.g., ethnicity).
Our framework significantly generalizes these notions by allowing arbitrary groups over which we may wish to impose constraints.

\begin{enumerate}[resume]
\item \textbf{Non-disjoint group constraints:} One can ensure both fully proportional representation by political party \emph{and} by demographic group by placing constraints (a) above on groups that correspond to the political party and constraints (b) on groups that correspond to demographics.
    Sometimes, groups can be arbitrary subsets instead of multiple partitions. For instance, let groups represent the major of applicants, and there are applicants with a double major or even more. Our framework can also handle this; see Example~\ref{ex:affectscore} in Section~\ref{sec:example}.
\end{enumerate}

\subsection{Feasibility Conditions and Properties}
\label{sec:feasibilityconditions}

Our algorithmic results can bypass the barriers posed by our hardness results in Section~\ref{sec:hard} by assuming the instances are feasible.
Here we argue that under many natural conditions in the multiwinner voting setting,  we can {\em deduce} that  feasible solutions exist and can construct them. We present some examples below.

\begin{enumerate}
\item \textbf{The proportional representative condition ensures feasibility:}
Assume that
    \begin{equation}
    \label{eq:klarge}
    k\geq 100\ln p.
    \end{equation}
This is natural as we expect $k$ to be much larger than $\ln p$; each candidate only has a small set of types over which we would like to impose fairness.
 Furthermore, assume that there is some slack between the true and allowable fraction of each type in the selected committee, i.e.,
    \begin{equation}
    \label{eq:proportion}
    \ell_i \leq k\cdot\left(\frac{|P_i|}{m}-c\right),~ u_i\geq k\cdot \left(\frac{|P_i|}{m}+c \right)
    \end{equation}
    for all $i\in [p]$ and some small constant $c\geq \sqrt{\nicefrac{3\ln p}{k}}$.
    Given these assumptions, a random committee uniformly chosen from all committees of size $k$ is feasible with high probability by the sampling variant of Chernoff bound \footnote{see \cite{DBLP:books/daglib/0012859} for details.} and union bound.

\item \textbf{The single type condition ensures feasibility:}
Assume that for each type $i\in [p]$, there are sufficiently many candidates with \emph{only} this type $\gamma_i=\left |j\in [m]:T_j=\left\{i\right\} \right |\geq \ell_i$. Moreover, we assume $\sum_{i\in [q]}\min\{\gamma_i,u_i\}\geq k$.
This guarantees that we can always select $k$ candidates with single type satisfying all fairness constraints.

\item \textbf{The bounded-parameter condition ensures feasibility:}
Assume we are in a setting where the $u_i$s are unbounded; this is natural, e.g., when we simply want to ensure minority representations.
Further, assume that $\sum_{i\in [p]} \ell_i \leq k$.
Thus, we can ensure the committee includes at lest $\ell_i$ candidates belonging to $P_i$.
\end{enumerate}

\subsection{Price of Fairness}
\label{sec:example}
Enforcing fairness must naturally come at a cost -- the feasible space of committees becomes smaller and hence the optimal score may decrease.
This leads to a natural question: {To what extent does the score decrease by introducing the fairness constraints?}

In some cases, the constraints can result in an arbitrarily bad score; see the following examples.
The first example shows that even a small change of fairness parameters may lead to a significant difference in the optimal score.

\begin{example}
\label{ex:affectscore}
Consider a multiwinner voting instance with $k=2$, $m=50$, and $n=200$. Let $P_1=\left\{c_1,c_3\right\}$, $P_2=\left\{c_2,c_3\right\}$, $P_3=\left\{c_1,c_4\right\}$, $P_4=\left\{c_2,c_4\right\}$, and $P_5=\left\{c_3,c_4\right\}$. Voters $a_1,\ldots,a_{100}$ have a preference order $c_1>c_2>c_5>c_6>\ldots>c_{50}>c_3>c_4$, and the other 100 voters $a_{101},\ldots,a_{200}$ have a preference order $c_2>c_1>c_5>c_6>\ldots>c_{50}>c_4>c_3$. Suppose we use the $\beta$-CC rule. The optimal committee of the unconstrained case is $\left\{c_1,c_2\right\}$ of total score 49*200=9800.
\begin{enumerate}
\item If the fairness parameters are $\ell_i=u_i=1$ $(1\leq i\leq 4)$, $\ell_5=1$, and $u_5=2$. Then there is only one feasible committee $\left\{c_3,c_4\right\}$ of total score 1*200=200. We lose a lot of total score in this setting.
\item However, if we slightly change the fairness parameters by resetting $\ell_5=0$ (other parameters keep the same), then $\left\{c_1,c_2\right\}$ is a feasible committee of score 9800, which is equal to the unconstrained case.
\end{enumerate}
\end{example}

\noindent
The second example shows that even a slight amount of flexibility in the constraints (i.e., $\ell_i < u_i$)  can significantly improve the score as compared to existing  ``fixed'' notions of fairness such as fully proportional representation and fixed degressive proportionality (in which, effectively, $\ell_i = u_i$).

\begin{example}
\label{ex:flex}
Consider a multiwinner voting instance with $k=20$, $m=100$, and $n=1000$. Let $P_1=\left\{c_1,\ldots,c_{10}\right\}$, and $P_2=\left\{c_{11},\ldots,c_{90}\right\}$. Voters are also partitioned into two groups: 100 voters prefer type 1, and 900 voters prefer type 2. The goalbase score function is defined as follows: A committee has score 100 if it has exactly three members in $P_1$, and has score 0 otherwise. The optimal score of the unconstrained case is 100, which can be achieved by $\left\{c_8,c_9,\ldots,c_{27}\right\}$.
\begin{enumerate}
\item Suppose we require full proportionality, i.e., the size-20 committee consists of $\frac{20\cdot 100}{1000}=2$ candidates in $P_1$ and $\frac{20\cdot 900}{1000}=18$ candidates in $P_2$. In this case, we set $\ell_1=u_1=2$ and $\ell_2=u_2=18$. Any committee satisfying this condition has score 0.
\item Suppose we require the Penrose square root law, i.e., the size-20 committee consists of $\frac{20\cdot \sqrt{100}}{\sqrt{100}+\sqrt{900}}=5$ candidates in $P_1$ and $\frac{20\cdot \sqrt{900}}{\sqrt{100}+\sqrt{900}}=15$ candidates in $P_2$. In this case, we set $\ell_1=u_1=5$ and $\ell_2=u_2=15$. Any committee satisfying this condition has score 0.
\item Suppose we only require a flexible proportionality: the selected committee contains at least two representations in $P_1$ and at least 15 representations in $P_2$. In this case, we set $\ell_1=2$, $\ell_2=15$ and $u_1=u_2=20$. An optimal committee is $\left\{c_8,c_9,\ldots,c_{27}\right\}$ of score 100.
\end{enumerate}
\end{example}

\noindent
On the other hand, in natural settings, e.g., when each voter prefers candidates of a given type $i$ over all other candidates, the optimal constrained score remains close to the optimal unconstrained score no matter how we select the fairness parameters; see the following example.

\begin{example}
\label{ex:scorelarge}
Consider a multiwinner voting instance satisfying the following conditions.
\begin{enumerate}
\item For each $i\in [p]$, we have $|P_i|\leq k$. Moreover, $p\leq k$.
\item Each voter $a$ has a preferred type $i$, and each candidate in $P_i$ is one of her favourite $k$ candidates.
\end{enumerate}
\noindent
Consider any feasible fairness constraints satisfying that $l_i\geq 1$ for all $i\in [p]$. We have the following observations.
\begin{enumerate}
\item Suppose we use the $\alpha$-CC rule. Since each group has at least one representation, the positional score of each voter is 1. Therefore, the optimal score must be $n$, which is equal to the unconstrained case.
\item Suppose we use the $\beta$-CC rule. The optimal score of the unconstrained case is at most $(m-1)n$. By the condition of fairness constraints, each group has at least one representation. Hence the positional score of each voter is at least $m-k$. Therefore, the optimal score is at least $(m-k)n$, no matter how we choose the fairness parameters.
\end{enumerate}
\end{example}

\noindent
Finally, we give an example illustrating that our framework can ensure fairness across multiple attributes and loss few score as compared to the unconstrained version.

\begin{example}
\label{ex:multivote}
Consider a multiwinner voting instance with $k=4$, $m=8$, and $n=200$. There are two attributes: gender and ethnicity, and four groups: men group $P_1=\left\{c_1,c_2,c_5,c_6\right\}$, women group $P_2=\left\{c_3,c_4,c_7,c_8\right\}$, Caucasian group $P_3=\left\{c_1,c_2,c_3,c_4\right\}$, and African-American group $P_4=\left\{c_5,c_6,c_7,c_8\right\}$. The preference orders are given as follows.
\begin{enumerate}
\item Voters $a_1,\ldots,a_{50}$ have a preference order $c_1>c_3>c_4>c_2>c_5>c_6>c_7>c_8$.
\item Voters $a_{51},\ldots,a_{100}$ have a preference order $c_2>c_4>c_3>c_1>c_5>c_6>c_7>c_8$.
\item Voters $a_{101},\ldots,a_{150}$ have a preference order $c_5>c_7>c_8>c_6>c_1>c_2>c_3>c_4$.
\item The last fifty voters $a_{151,\ldots,200}$ have a preference order $c_6>c_8>c_7>c_5>c_1>c_2>c_3>c_4$.
\end{enumerate}
\noindent
Suppose we use the $\beta$-CC rule. The optimal committee of the unconstrained case is $\left\{c_1,c_2,c_5,c_6\right\}$ of total score 7*200=1400. However, this committee consists of four men and lacks fairness in gender.

Our framework can overcome this problem by setting $\ell_i=u_i=2$ for all $1\leq i\leq 4$, i.e., we require the chosen committee to consist of two men and two women, while consisting of two Caucasians and two African-Americans. Then an optimal committee is $\left\{c_1,c_4,c_5,c_8\right\}$ of total score 7*100+6*100=1300.
Thus by introducing the fairness constraints, the optimal committee is fair both of gender and ethnicity and loses only a small amount of score.
\end{example}


\subsection{Properties of Voting Rules with Fairness Constraints}
\label{sec:properties}
There are various properties that generally one would like a multiwinner voting rule to satisfy.
For example, Elkind et al. \cite{elkind2014properties} show that the single non-transferable vote (SNTV) rule satisfies many properties including committee monotonicity, solid coalitions, consensus committee, weak unanimity, monotonicity, homogeneity, and consistency.
A variety of nice properties continue to be satisfied by the scoring rules in the presence of fairness constraints. We present some examples below.

Consistency means that if $S$ is an optimal committee
for voters $A_1$ and also for voters $A_2$, then $S$ must also be an optimal committee for the set of voters $A_1\cup A_2$.
As the fairness constraints restrict the set of feasible committees in the same way for any set of voters, the argument for consistency remains the same (see \cite[Theorem 7]{elkind2014properties}).
In fact, consistency holds for any committee scoring rule with fairness constraints.
Similarly, by the same arguments as in \cite{elkind2014properties}, if a committee scoring rule satisfies monotonicity (respectively, homogeneity)
then the same scoring rule in the presence of fairness constraints still satisfies monotonicity (respectively, homogeneity).

On the other hand, the remaining properties
may not be preserved.
The problem arises because the fairness constraints can make certain committees infeasible and thus force the desired property to be violated.
Feasibility, however, appears to be the only bottleneck and motivates the definition of corresponding fair properties.
For example, weak unanimity states that if a set of $k$ candidates $S$ dominates (see \cite[Section 2.2.1]{faliszewski2017multiwinner}) all other size $k$ subsets of candidates with respect to any voter's preference list, then $S$ must be an optimal committee.
We can instead define a fair version of weak unanimity; see the following definition.

\begin{defn} (Fair weak unanimity.) Let $\calB$ be the collection of all committees of size $k$ that satisfy all fairness constraints. If $S\in \calB$ dominates (see \cite[Section 2.2.1]{faliszewski2017multiwinner} for the definition)
any other committee in $\calB$ with respect to any voter's preference list,
then $S$ must be an optimal committee.
\end{defn}

\noindent
Any committee scoring rule that satisfies such weak unanimity (or, similarly, committee monotonicity) will satisfy the fair version of these properties in the presence of fairness constraints.

\subsection{Incorporating Fariness in Multiwinner Voting Rules}
\label{sec:incorporating}
Some multiwinner voting rules are not defined by a score function; e.g., the single transferable vote (\emph{STV}) rule \cite{tideman1995single} and the {\em Greedy} Monroe rule \cite{skowron2015achieving}, which select candidates in rounds instead of selecting the whole committee simultaneously.
It is unclear how fairness constraints can be added to such voting rules and we leave this problem for future work.

\section{Other Related Work}
\label{sec:related}

The study of total score functions and their resulting optimization problems have received much attention in recent years.
Often the optimization problem turns out to be NP-hard; both $\alpha$-CC and $\beta$-CC are NP-hard~\cite{procaccia2008complexity},
$1-1/e$ is the best approximation ratio for $\alpha$-CC~\cite{skowron2015achieving},
the Monroe rule is computationally hard even if the voting parameters are small~\cite{betzler2013computation},
and the OWA-based rules are hard in general ~\cite{skowron2016finding}
as are the goalbase rules in various settings~\cite{uckelman2009representing}.
Hence, one must largely resort to developing approximation algorithms for these problems.
Towards this, there has been a rich line of work \cite{lu2011budgeted,skowron2015achieving,skowron2015achieving,skowron2016finding}.
The majority of score functions for multiwinner voting rules that have been studied are monotone submodular.
Algorithm design for such score functions
have benefitted from theoretical developments in the area of monotone submodular function maximization
\cite{nemhauser1978analysis,calinescu2011maximizing}.

Some recent work also considers aspects of fairness in voting.
Goalbase score functions, which specify an arbitrary set of logic constraints and let the score  capture the number constraints satisfied~\cite{uckelman2010alice,uckelman2009representing}, could be used to ensure fairness.
However, there are no known efficient algorithms to solve goalbase functions.
{Some recent literature studies single-winner voting in the multi-attribute setting; see the survey of Lang and Xia~\cite{DBLP:reference/choice/LangX16}.
	The bi-apportionment model can handle up to two attributes (often political party and district)~\cite{serafini2012parametric,lari2014bidimensional}.
	Another related model is called constraint approval voting (CAP), with constraints on the numbers of winners from different categories of candidates, proposed by Brams~\cite{brams1990constrained} and Potthoff~\cite{potthoff1990use}.
	However, there is no efficient algorithm since the input in CAP is exponentially large in the number of attributes.
}
Lang and Skowron~\cite{DBLP:conf/aaai/LangS16} also consider the problem of committee selection with multiple partitions; however, the goal is to produce a committee close to a given target composition of attributes as opposed to maximizing a score function.

More generally, our results contribute to the growing set of algorithms that incorporate fairness constraints to counter algorithmic bias in fundamental algorithmic problems such as classification \cite{dwork2012fairness,zemel2013learning,zafar2017fairness,zafar2017fair}, 
sampling \cite{FatML,celis2017complexity,celis2018fair}, ranking \cite{celis2018ranking,yang2017measuring} and personalization \cite{CV17}. 
Independently of this paper, \cite{bredereck2017multiwinner}  also propose a  model for multiwinner voting with a type of fairness constraints (referred to as \emph{diversity constraints}).

\section{Hardness Results}
\label{sec:hard}

\noindent
In this section we present our hardness results for the constrained multiwinner voting problem.
Our first theorem addresses the complexity of the feasibility problem. Recall that $\Delta$ is the maximum number of groups in which a candidate can be.
The following theorem shows that when a candidate may be part of $3$ or more groups, just the feasibility problem can become NP-hard; even under mild feasibility conditions.
\begin{theorem}
\label{thm:Delta3hard} (NP-hardness of feasibility: $\Delta\geq 3$)
The constrained multiwinner voting feasibility problem is NP-hard for any $\Delta\geq 3$. This is true even if we assume that all $\ell_i=0$ or all $u_i= |P_i|$.
\end{theorem}

\begin{proof}
For the case that all $u_i= |P_i|$ and $\Delta\geq 3$, we reduce from $\Delta$-hypergraph matching, which is known to be NP-hard \cite{hazan2003complexity}. The reduction is inspired by the NP-hard argument for the constrained ranking feasibility problem \cite[Theorem 3.1]{celis2018ranking}.
Let $G=(V,E)$ be a $\Delta$-hypergraph with $|V|=p$ and $|E|=m$. We construct a constrained multiwinner voting instance as follows. For each hyperedge $e_i\in E$, we construct a corresponding candidate $c_i$. For each vertex $v_i\in V$, we construct a corresponding group $P_i=\left\{c_j:v_i\in e_j\right\}$. Thus, there are $p$ groups and $m$ candidates.
For each group $P_i$, let $\ell_i=0$ and $u_i=1$, meaning that each group can be hit at most once. It corresponds to the requirement that each vertex can be covered by at most one hyperedge. %
Hence $G$ admits a matching of size $k$ if and only if there is a feasible committee of size $k$ in the constrained multiwinner voting instance, which completes the proof.

For the case that all $u_i= |P_i|$ and $\Delta\geq 3$, we present a reduction from the 3-regular vertex cover problem, which is known to be NP-hard \cite{DBLP:conf/ciac/AlimontiK97}. Let $G=(V,E)$ be a 3-regular graph with $|V|=m$ and $|E|=p$. The problem is to test whether there is a vertex cover of size $k$.
We construct a constrained multiwinner voting instance whose feasibility is equivalent to $G$ having a vertex cover of size $k$ as follows.
We construct a corresponding candidate $c_i$ for each vertex $v_i\in V$, and a corresponding group $P_i=\left\{c_j,c_{j'}\right\}$ for each edge $e_i=(v_j,v_{j'})$. Thus, there are $m$ candidates, $p$ groups and the degree $\Delta=3$.
Let $\ell_i=1$ for all $i\in [p]$, i.e., each group has at least one representation. It corresponds that each vertex in $G$ should be covered by some chosen edge.
Therefore, it can be concluded that each vertex cover of size $k$ corresponds to a feasible committee of size $k$, which finishes the proof.
\end{proof}

\noindent
The next two theorems show that the feasibility problem remains hard even if one allows the size of the committee or the fairness constraints to be violated.
\begin{theorem} (Hardness of feasibility with committee violations)
\label{thm:lowerhard}
Let $u_i= |P_i|$ for all $i\in [p]$. For any  $\varepsilon>0$, the following  gap violation variant of the constrained multiwinner voting feasibility problem is NP-hard.
\begin{enumerate}
\item Output YES if the input instance is feasible.
\item Output NO if there is no feasible solution of size less than $(1-\varepsilon)k\ln p$.
\end{enumerate}
\end{theorem}
%

\begin{proof}
We make a reduction from the following constrained set multi-cover problem \cite{berman2007randomized}, which is NP-hard to approximate within a factor $(1-\varepsilon)\ln p$ by \cite{berman2007randomized}. Given a ground set $X=\{v_1,\ldots,v_p\}$ together with a weight function $w:X\rightarrow \Z_{\geq 0}$, and a collection of $m$ sets $S_i\subseteq X$, the goal is to choose some sets of minimum cardinality covering each element $v_i$ by a factor of at least $w_i$.

Then we construct a constrained multiwinner voting instance. Construct a corresponding candidate $c_i$ for each set $S_i$, and a corresponding group $P_i=\{c_j:v_i\in S_j\}$ for each element $v_i$. Hence there are $m$ candidates and $p$ groups. Define $\ell_i$ to be $w_i$ for all $i\in [p]$. Observe that a feasible committee of size $k$ exists if and only if the minimum cardinality of the constrained set multi-cover problem is at most $k$.
Therefore, a polynomial-time algorithm for the constrained multiwinner voting feasibility problem implies a $(1-\varepsilon)\ln p$-approximation algorithm for the constrained set multi-cover problem. This fact completes the proof.
\end{proof}

\begin{theorem} (Hardness of feasibility with fairness  violations)
\label{thm:constantviolation}
Assume $\ell_i=0$ for each $i\in [p]$. For every  $\theta > 1$, the following violation variant of the constrained multiwinner voting feasibility problem is NP-hard.
\begin{enumerate}
\item Output YES if the input instance is feasible.
\item {Output NO if there is no solution which violates every fairness constraint by a factor of at most $\theta$, i.e., there is no solution $S$ of size $k$ such that $|S\cap P_i|\leq \theta u_i$.}
\end{enumerate}
\end{theorem}

\begin{proof}
Similar to \cite[Theorem 3.4]{celis2018ranking}, we use the inapproximability of independent set \cite{Hastad96,zuckerman2006linear} to prove the theorem. It is NP-hard to approximate the cardinality of the maximum independent set problem in undirected graphs within a factor of $|V|^{1-\varepsilon}$ for any constant $\varepsilon>0$.

Given a graph $G=(V,E)$ with $|V|=p$ and $|E|=m$ and a number $k$, the goal is to check whether there exists a independent set of size $k$ in $G$.
For each vertex $v_i\in V$, we construct a corresponding candidate $c_i$. For every cardinality-$\theta$ clique, we construct a property and set its fairness upper bound to be 1. Observe that there are at most $m^{\theta+1}=poly(m)$ fairness constraints.

We have the following claim:
\begin{enumerate}
\item If there is a committee of size $k$ that violates every fairness constraint by a factor of at most $\theta$, then there is an independent set of cardinality $\Omega(k^{1/(\theta+1)})$ in $G$.
\item If there is no feasible committee of size $k$, then there is no cardinality-$k$ independent set.
\end{enumerate}
\noindent
If the above claim holds, then a polynomial-time algorithm for the constrained multiwinner voting feasibility problem implies a $|V|^{1-1/(\theta+1)}$-approximation algorithm for the maximum independent set problem. Hence it remains to prove the claim.

For the first claim, if a committee of size $k$ that violates every fairness constraint by a factor of by a factor of at most $\theta$ exists, then we have a subset $S\subseteq V$ of size-$k$ that does not contain any $(\theta+1)$-clique. By a standard upper-bound on the Ramsey number $R(k^{1/(\theta+1)},\theta)$, there exists an independent set of cardinality $\Omega(k^{1/(\theta+1)})$ in $G$.
The second claim is not hard, since every independent set of size $k$ implies a feasible committee of size $k$ for the constrained multiwinner voting instance.
\end{proof}

\noindent
Further, we show that even if we assume that we know a feasible instance of the constrained multiwinner voting problem, the hardness does not go away.
\begin{theorem} (Inapproximability for feasible instances)
\label{thm:approxhard}
A feasible constrained multiwinner voting problem that satisfies $\ell_i=0$ for each $i\in [p]$ is NP-hard to approximate within a factor of $\Omega(\log \Delta/\Delta)$.
\end{theorem}

\begin{proof}
It suffices to prove the SNTV case. The reduction is from maximum $\Delta$-hypergraph matching, which is hard to approximate within a factor of $\Omega(\log \Delta/\Delta)$ by \cite{hazan2003complexity}.

Let $G=(V,E)$ be a $\Delta$-hypergraph with $|V|=p$ and $|E|=m$.
We use the same construction of the constrained multiwinner voting instance as in the proof of Theorem~\ref{thm:Delta3hard}. In addition, we construct another $k$ candidates $c'_{1},\ldots,c'_{k}$ which do not belong to any group. Note that this instance is satisfiable since $\left\{c'_1,\ldots,c'_k\right\}$ is a feasible committee.
Then we define the total score function. We first construct $m$ voters and let voter $a_i$ most prefer to $c_i$ ($i\in [m]$). According to the SNTV rule, a feasible committee consisting of $k_1$ candidates from $\left\{c_i\right\}$ and $k-k_1$ candidates from $\left\{c'_i\right\}$ has score $k_1$. On the other hand, such a committee corresponds to a cardinality-$k_1$ matching of the maximum $\Delta$-hypergraph matching problem.
Therefore, to compute the optimal score is equivalent to finding the maximum $\Delta$-hypergraph matching in $G$. Since the reduction is approximation preserving, we finish the proof.
\end{proof}


\section{Algorithmic Results}
\label{sec:alg}
In this section we present our main algorithmic results. We first introduce a useful notion called ``multilinear extension" for monotone submodular optimization.

\begin{defn} (Multilinear extension)
\label{def:multilinear}
Given a monotone submodular function $f:\{0,1\}^m\rightarrow \R_{\geq 0}$,  the multilinear extension $F:[0,1]^m\rightarrow \R_{\geq 0}$ is defined as follows:
For $y=\left(y_1,\ldots, y_m\right)\in [0,1]^m$, denote $\hat{y}$ to be a random vector in
$\left\{0,1\right\}^m$ where the $j$th coordinate is independently rounded to 1 with probability $y_j$ or 0 otherwise. Then we let
$$
F(y)=\Exp\left[f(\hat{y})\right]=\sum_{R\subseteq [m]} f(R) \prod_{i\in R} y_i \prod_{j\notin R} (1-y_j).
$$
\end{defn}  \noindent
\noindent Let $\e_i=(0,\ldots,0,1,0,\ldots,0)\in \{0,1\}^{m}$.
\begin{lemma}
\label{claim:convex}
\cite{calinescu2011maximizing}
For any $y\in [0,1]^m$ and $i,j\in [m]$,  $F^y_{ij}(t):=F\big(y+t(\e_i-\e_j)\big)$ is convex.
\end{lemma}

\noindent
Define $\calB\subseteq 2^{[m]}$ as the family of all committees of size $k$ that satisfy the fairness constraints.
Denote a polytope $B:=\big\{y\in  [0,1]^m $ $ \mid \sum_{i\in [m]}y_i=k; \ell_j\leq \sum_{i\in P_j} y_i\leq u_j, \forall j\in [m]\big\}$ to be the set of all vectors that satisfy the cardinality constraint and all fairness constraints.
Let $F:[0,1]^m\rightarrow \R_{\geq 0}$ denote the multilinear extension of the total score function $\score$. Let $OPT$ be the optimal score of the constrained MS multiwinner voting problem.
%

\subsection{The Case of $\Delta\geq 2$}
\label{sec:degree2}

Theorem~\ref{thm:approxhard} implies that it may be hard to find a committee only violating the fairness constraints by a small amount when $\Delta \geq 3$.
On the other hand, the following theorem shows that a constant approximation solution can be achieved that violates  all fairness constraints by at most a multiplicative factor for feasible instances.
By taking the violation factors into account, the desired fairness can be achieved by setting tighter constraints.

\begin{theorem} (Bi-criterion algorithm when  $\Delta\geq 2$)
\label{thm:general}
Consider a feasible constrained MS multiwinner voting instance with $OPT\gg \Omega(1)$.
Let $L:=\min_{i\in [p]}\ell_i$ and $U:=\min_{i\in [p]}u_i$.
Assume $\nicefrac{2\sqrt{\ln p}}{\sqrt{U}}\leq 1$.
There exists a randomized polynomial-time algorithm that outputs a committee $S$ of size $k$ with score $(1-1/e-o(1))OPT$ with constant probability, and $S$ satisfies the following for all $i \in [p]$:
\begin{equation}
\label{eq:bi-criterion}
\left(1-\nicefrac{2\sqrt{\ln p}}{\sqrt{L}}\right)\ell_i  \leq |S\cap P_i| \leq \left(1+\nicefrac{2\sqrt{\ln p}}{\sqrt{U}}\right)u_i.
\end{equation}
The approximation ratio is $1-o(1)$ for the SNTV rule.
\end{theorem}

\noindent
Before proving this theorem, we discuss the assumption and the consequences of Theorem~\ref{thm:general}.
Firstly, the assumption that $OPT\gg \Omega(1)$ is reasonable for many voting rules, such as the CC rule, the OWA-based rule, and the Monore rule.
The reason is that if \emph{enough} (say at least $n/10$) voters have at least one representation in the optimal committee, then the total score of these rules is usually at least $n/10$.

Under reasonable assumptions, the violation in fairness constraints in the above theorem can be seen to be small.
First,  assume that no group is  too small:
$|P_i|\geq 0.15m \quad \forall i\in [m].$
Groups corresponding to gender, ethnicity and political opinions are often large.
Combining this with the proportional representative condition ({see Equations \eqref{eq:klarge} and \eqref{eq:proportion} in Section~\ref{sec:feasibilityconditions}}):
$$k\gg 100\ln p ,$$ and $$\ell_i \approx k\cdot\left(\frac{|P_i|}{m}-0.05\right),~ u_i\geq k\cdot \left(\frac{|P_i|}{m}+0.05 \right)$$
for all $i\in [p]$,
we observe that $\nicefrac{2\sqrt{\ln p}}{\sqrt{L}}\ll 0.66$ is a small number.
Thus, the violation of  the group-fairness condition by Theorem~\ref{thm:general} is small.
{We expect such algorithmic solutions to be deployed for the development of automated systems, such as movie selection on the airplane and news recommendation for websites, for which the above assumptions are natural.}

\begin{proof} [of Theorem~\ref{thm:general}]
Let $\varepsilon>0$ be any given constant.
Recall that $B:=\big\{y\in  [0,1]^m $ $ \mid \sum_{i\in [m]}y_i=k; \ell_j\leq \sum_{i\in P_j} y_i\leq u_j, \forall j\in [m]\big\}$ is the set of all vectors that satisfy the cardinality constraint and all fairness constraints.
We first obtain a fractional solution $y\in B$ by the continuous greedy algorithm in~\cite[Section 3.1]{calinescu2011maximizing}.
It follows that
$F(y)\geq (1-1/e)OPT\gg \Omega(1)$ by~\cite[Appendix A]{calinescu2011maximizing}.
Next, we run the randomized swap rounding algorithm in~\cite{chekuri2010dependent}
and obtain a committee $S$ of size $k$.
By~\cite[Theorem 2.1]{chekuri2010dependent}, we have
$\Exp\left[\score(S)\right]\geq F(y)\geq (1-1/e)OPT,$
and for any $i\in [p]$ and any $\delta_1, \delta_2>0$,
\begin{eqnarray}
\label{eq:probability}
\begin{split}
& \Prob\left[|S\cap P_i|\leq (1-\delta_1)\ell_i\right]\leq e^{-\ell_i\delta_1^2/2}, \\
&\Prob\left[|S\cap P_i|\geq (1+\delta_2)u_i\right]\leq \left(\nicefrac{e^{\delta_2}}{(1+\delta_2)^{1+\delta_2}}\right)^{u_i}.
\end{split}
\end{eqnarray}
If $\delta_2\leq 1$, we have the following inequality.
\begin{eqnarray}
\begin{split}
&\Prob\left[|S\cap P_i|\geq (1+\delta_2)u_i\right] & \\
\leq& \left(\nicefrac{e^{\delta_2}}{(1+\delta_2)^{1+\delta_2}}\right)^{u_i}& \text{(Ineq. \eqref{eq:probability})} \\
\leq& \left(\nicefrac{e^{\delta_2}}{e^{\delta_2+0.38\delta^2}}\right)^{u_i} & \text{(Claim~\ref{cl:0.38})}\\
=& e^{-0.38 u_i \delta_2^2} &.
\end{split}
\label{eq:upper}
\end{eqnarray}
Here, the second inequality is from the following claim.

\begin{claim}
\label{cl:0.38}
For any $x\in [0,1]$, $x+0.38x^2\leq \left(1+x\right) \ln \left(1+x\right)$.
\end{claim}

\begin{proof}
Let $f(x)=\left(1+x\right) \ln \left(1+x\right)-x-0.38x^2$.
It is equivalent to prove that $f(x)\geq 0$ for all $x\in [0,1]$.
We have
$$f'(x)=\ln \left(1+x\right)-0.76x,$$
and
$$f''(x)=\frac{1}{1+x}-0.76.$$
Observe that $f''(x)> 0$ when $x\in (0,6/19)$ and $f''(x)<0$ when $x>6/19$.
Hence we conclude that $f'(x)$ is monotone increasing in $[0,6/19]$ and is monotone decreasing in $[6/19,1]$.
Also note that $f'(0)=0$ and $f'(1)=\ln 2-0.76<0$.
Then there exists only one $\theta\in (0,1)$ such that $f'(\theta)=0$.
Therefore, $f(x)$ is monotone increasing in $[0,\theta]$ and is monotone decreasing in $[\theta,1]$.
Note that $f(0)=0$ and $f(1)=2\ln 2-1.38>0$.
We conclude that $f(x)\geq 0$ for all $x\in [0,1]$.
\end{proof}

\noindent
Let $\delta_1=\nicefrac{2\sqrt{\ln p}}{\sqrt{L}}$ and $\delta_2=\nicefrac{2\sqrt{\ln p}}{\sqrt{U}}$.
By the assumption of the theorem, we have $\delta_2\leq 1$.
By the union bound, we have the following inequality
\begin{eqnarray*}
& &\Prob \big[ \forall i\in [p], \left(1-\nicefrac{2\sqrt{\ln p}}{\sqrt{L}}\right)\ell_i  \leq |S\cap P_i|\leq  \\
& &\left(1+\nicefrac{2\sqrt{\ln p}}{\sqrt{U}}\right)u_i \big] \\
&\geq & 1-\sum_{i\in [p]}  \Prob\left[|S\cap P_i|\leq \left(1-\nicefrac{2\sqrt{\ln p}}{\sqrt{L}}\right)\ell_i\right] \\
& &-\sum_{i\in [p]}  \Prob\left[|S\cap P_i|\geq \left(1+\nicefrac{2\sqrt{\ln p}}{\sqrt{U}}\right)u_i\right]  \\
& &\text{(union bound)} \\
&\geq & 1-\sum_{i\in [p]} e^{-\ell_i \left( \nicefrac{2\sqrt{\ln p}}{\sqrt{L}}\right)^2/2}-\sum_{i\in [p]} e^{-0.38 u_i  \left( \nicefrac{2\sqrt{\ln p}}{\sqrt{U}}\right)^2}  \\
&  &\text{(Inequalities~\eqref{eq:probability} and~\eqref{eq:upper})}\\
&\geq & 1-\sum_{i\in [p]} e^{-2\ell_i \ln p/L}-\sum_{i\in [p]} e^{-1.5 u_i \ln p/U} \\
&\geq & 1-\sum_{i\in [p]} e^{-2\ln p}-\sum_{i\in [p]} e^{-1.5 \ln p}  \\
& & \text{(Definitions of $L, U$)}\\
&=& 1-p \cdot \frac{1}{p^2}-p \cdot \frac{1}{p^{1.5}} \geq 1-2/\sqrt{p}.
\end{eqnarray*}
which shows that $S$ satisfies Equation \eqref{eq:bi-criterion} with probability at least $1-2/\sqrt{p}$.
On the other hand, we have for any $\delta>0$,
$
\Prob\left[\score(S)\leq (1-\theta)F(y)\right]\leq e^{-F(y)\theta^2/8}
$
by~\cite[Theorem 2.2]{chekuri2010dependent}.
Let $\theta=0.01$.
Since $F(y)\gg \Omega(1)$, we have $\score(S)\geq (1-1/e-0.01)OPT$ with probability $1-o(1)$. Combining all of the above, along with a union bound, we conclude the proof of the theorem.

Due to the fact that the total score function in the SNTV rule is linear, we obtain a fractional solution with $F(y)\geq OPT$ in the first stage.
Then by the same argument, the resulting committee $S$ satisfies that $\score(S)\geq (1-o(1))OPT$.
\end{proof}

\noindent
We now present our algorithmic results when there are only one-sided (upper or lower) fairness constraints.
\noindent
We first consider the case that all $\ell_i=0$, i.e., each group is only required not to be over-represented.
In this case, we call an algorithm $(\gamma,\theta)$-approximation ($\gamma \in [0,1], \theta \geq 1$) if the algorithm outputs a $\gamma$-approximate solution $S$ which violates all constraints by a factor of at most $\theta$, i.e., $|S\cap P_i|\leq \theta \cdot u_i$ holds for all $i\in [p]$.
The assumption that we have a feasible solution $\hat{S}$ is often satisfied and does not add computational overhead under natural conditions on the data; see Section~\ref{sec:dis} for examples. 

\begin{theorem} (Bi-criterion algorithm when $\ell_i=0$)
\label{thm:onlyupper}
Consider the constrained MS multiwinner voting problem satisfying $\ell_i=0$ for all $i\in [p]$. Suppose we have a feasible solution $\hat{S}$ in advance. For any constant $\varepsilon>0$, the following claims hold
\begin{enumerate}
\item There exists a $(\nicefrac{1}{\Delta+1},2)$-approximation algorithm that runs in $O(mk/\Delta)$ time.
\item For any constant $0<\varepsilon<1$, suppose that $u_i\geq 6\log p/\varepsilon^2$ and $k\geq 6\log p/\varepsilon^2$. There exists a polynomial-time $(1-1/e-O(\varepsilon),2)$-approximation algorithm.
\end{enumerate}
\end{theorem}
\noindent
\begin{proof}
Define $\calB'$ to be the collection of committees that have  at most $k$ candidates and satisfy all fairness constraints.
For the first claim, we first prove that $(C,\calB')$ is a $\Delta$-extendible system
\footnote{A pair $(\calN,\calI)$, where $\calN$ is a finite set and $\calI$ is a collection of subsets of $\calN$, is called a $\Delta$-extendible system if 1) (downclosed) $S\subseteq T\subseteq \calN$ and $T\in \calI$ imply $S\in \calI$; 2) ($\Delta$-extendible) For any $S,T\in \calI$ and any element $e\in \calN\setminus S$, if $S\subseteq T$ and $S\cup \{e\}\in \calI$, then there must exist a subset $Y\subseteq T\setminus S$ with $|Y|\leq \Delta$ such that $T\cup\{e\}\setminus Y \in \calI$.}.
It is not hard to see the downclosed property holds.
For the $\Delta$-extendible property, the proof is as follows.
For any $S,T\in \calB'$ and any candidate $c\in C$, let $P(c)$ denote the collection of types that $c$ has.
We have $|P(c)|\leq \Delta$.
Observe that $T'=T\cup \left\{e\right\}$ has at most $k+1$ candidates and $|T' \cap P_i|\leq 1+u_i$ holds for all $i\in P(c)$.
Then if $S\subseteq T$ and $S'=S\cup \left\{c\right\} \in \calB'$, we have $|S' \cap P_i|\leq u_i$ holds for all $i\in P(c)$.
It implies that if $|T' \cap P_i|= 1+u_i$ for some $i\in P(c)$, then $\left(T\setminus S\right)\cap P_i\neq \emptyset$.
Let $c_i\in \left(T\setminus S\right)\cap P_i$ be an arbitrary candidate for any $i\in P(c)$ satisfying that $|T' \cap P_i|= 1+u_i$.
Define $Y$ to be the collection of all such candidates $c_i$.
If $Y=\emptyset$, then let $Y=\left\{c'\right\}$ for an arbitrary candidate $c'\in T\setminus S$.
By the construction of $Y$, we conclude that $1\leq |Y|\leq \Delta$.
Hence $|T\cup\{c\}\setminus Y|\leq k$.
Moreover, for any $i\in P(c)$, we have $\left|\left(T\cup\{c\}\setminus Y\right)\cap P_i\right|\leq 1+u_i-1=u_i$, which implies that $T\cup\{c\}\setminus Y\in \calB'$.
Thus, we prove that $(C,\calB')$ is a $\Delta$-extendible system.

Since $\score$ is monotone submodular, we reduce the problem of finding $\arg\max_{S\in \calB'}\score(S)$ to the monotone submodular maximization problem with $\Delta$-extendible system, which has a $(\nicefrac{1}{\Delta+1})$-approximation algorithm in $O(mk/\Delta)$ time by~\cite[Theorem 1]{feldman2017greed}.
Therefore, we can compute a committee $S_1\in \calB'$ with $\score(S_1)\geq (\nicefrac{1}{\Delta+1}) OPT$ in $O(mk/\Delta)$ time by~\cite[Theorem 1]{feldman2017greed}. Since $S_1\in \calB'$, we have that $|S_1|\leq k$ and $S_1$ satisfies all fairness constraints.
By the condition of the theorem, there exists a feasible solution $\hat{S}$ in advance, i.e., $|\hat{S}|=k$ and $\hat{S}$ satisfies all fairness constraints. Then from $\hat{S}\setminus S_1$, we arbitrarily select a set $S_2\subseteq \hat{S}$ of $k-|S_1|$ candidates. Since both $S_1$ and $S_2$ satisfy all fairness constraints, we have for all $i\in [p]$,
\[
\left| \left(S_1\cup S_2\right) \cap P_i \right|\leq \left| S_1  \cap P_i \right|+\left| S_2  \cap P_i \right|\leq u_i+u_i=2u_i.
\]
By monotonicity, we have $\score\left(S_1\cup S_2\right)\geq \score(S_1)\geq (\nicefrac{1}{\Delta+1}) OPT$. Therefore, we conclude that the committee $S_1\cup S_2$ is a $(\nicefrac{1}{\Delta+1},2)$-approximation solution.

For the second claim,
we compute a fractional solution constrained to $(1-\varepsilon)\calB'$ defined as follows: $(1-\varepsilon)\calB':=\left\{ S\subseteq C: |S|\leq (1-\varepsilon)k; |S\cap P_i|\leq (1-\varepsilon)u_i, \forall i\in [p] \right\}$. Then we round it to a committee $S_1$ using the swap randomized rounding procedure in~\cite{chekuri2010dependent}.
By~\cite[Theorem 5.2]{chekuri2010dependent}, $S_1$ can be guaranteed to be in $\calB'$,
and $\score(S_1)\geq (1-1/e-O(\varepsilon)) OPT$.
Then by the same argument as in the first claim, we select $k-|S_1|$ more candidates from $\hat{S}\setminus S_1$ and obtain a $(1-1/e-O(\varepsilon),2)$-approximation solution.
\end{proof}

\noindent
We now present our result for the case when we only have lower bound constraints, i.e., the upper bound constraints $u_i = |P_i|$ for all $i\in [p]$ (observe that $|S\cap P_i|\leq u_i$ always holds in this case).
Further, we assume that a committee of size $o(k/\ln \Delta)$ satisfying all fairness constraints exists,
\noindent
which is reasonable since, in practice, $\Delta$ roughly represents the number of attributes, like gender and ethnicity.
Though there may exist many groups, the number of attributes is usually limited.

\begin{theorem} (Approximation algorithm when $u_i= |P_i|$)
\label{thm:onlylower}
Given a constrained multiwinner voting instance satisfying $u_i= |P_i|$ for all $i\in [p]$ with a promise that a committee of size $o(k/\ln \Delta)$ satisfying all fairness constraints exists, there exists a $(1-1/e-o(1))$-approximation algorithm.
\end{theorem}

\begin{proof}
Define $\calB'=\left\{S\subseteq C \mid |S\cap P_i|\geq \ell_i,~ \forall i\in [p] \right\}$ to be the collection of committees satisfying all fairness constraints.
Consider the  problem $\min_{S\in \calB'} |S|$ and its reduction to the minimum constrained set multi-cover problem.
For each group $P_i$, we construct a corresponding element $v_i$.
For each candidate $c_i\in C$, we construct a set $S_i=\{v_j:c_i\in P_j\}$.
Hence there are $p$ elements and $m$ sets, and each set has at most $\Delta$ elements.
The goal is to select the smallest number of sets that cover each element $v_i$ at least $\ell_i$ times.
By~\cite{berman2007randomized}, there exists a polynomial-time $(1+\ln \Delta)$-approximation algorithm for the minimum constrained set multi-cover problem.
Therefore, we can compute a committee $S_1$ of size $o(k)$ satisfying all fairness constraints since a committee of size $o(k/\ln \Delta)$ satisfying all fairness constraints exists.
The next step, we construct an unconstrained multiwinner voting problem:
Find a committee of size at most $(k-|S_1|)$ which maximizes the total score.
Suppose the optimal total score of this problem is $OPT'$. Since $|S_1|=o(k)$ and $\score$ is a monotone submodular function, we can prove that $OPT'\geq (1-o(1))OPT$ by reduction.
W.l.o.g., let $O=\left\{c_1,\ldots,c_k \right\}$ be the optimal committee in $\calB$, i.e., $\score(O)=OPT$.
Assume the optimal total score of a committee $S\subseteq O$ of size at most $i$ is $OPT_i$.
It suffices to prove that $OPT_i\geq \frac{i}{k}OPT$.
If $i=1$, since $\score$ is a nonnegeative monotone submodular function, we have
\begin{eqnarray}
\label{eq:monotone}
\begin{split}
& \sum_{j\in [k]}\score(\left\{c_j\right\}) & \\
\geq & \sum_{j\in [k-2]}\score(\left\{c_j\right\})+\score(\left\{c_{k-1},c_k\right\})+\score(\emptyset) & (\text{$\score$ is monotone submodular})\\
\geq & \sum_{j\in [k-3]}\score(\left\{c_j\right\})+\score(\left\{c_{k-2},c_{k-1},c_k\right\})+\score(\emptyset) +\score(\emptyset) & \\
\geq &\ldots &\\
\geq & \score(O)+(k-1)\score(\emptyset) &\\
\geq & OPT.&
\end{split}
\end{eqnarray}
Then there must exist $j\in [k]$ such that $\score(c_j)\geq OPT/k$.
By the definition of $OPT_1$, we have $OPT_1\geq \score(c_j)\geq OPT/k$.
If $OPT_i\geq \frac{i}{k}OPT$ holds, consider the case of $i+1$.
By the definition of $OPT_i$, there exists a committee $S\subseteq O$ of size at most $i$ such that $\score(S)=OPT_i$.
By a similar argument as in Equation \eqref{eq:monotone}, we have
\[
 \sum_{j\in O\setminus S}\left(\score(S\cup\left\{c_j\right\})-\score(S) \right) \geq OPT-\score(S).
\]
Hence there exists a candidate $c_j\in O\setminus S$ such that
$$\score(S\cup\left\{c_j\right\})-\score(S)\geq \frac{OPT-\score(S)}{k-i}.$$
Then we conclude that
\begin{eqnarray*}
& \score(S\cup\left\{c_j\right\})& \\
\geq & \score(S)+\frac{OPT-\score(S)}{k-i}& \\
= & \frac{OPT}{k-i}+\left(1-\frac{1}{k-i} \right) \score(S)& \\
\geq & \frac{OPT}{k-i}+\left(1-\frac{1}{k-i} \right) \frac{i}{k}OPT & (\text{Definition of $S$})\\
= & \frac{i+1}{k} OPT.&
\end{eqnarray*}
Thus, $OPT_i\geq  \score(S\cup\left\{c_j\right\})\geq \frac{i+1}{k} OPT$ which completes the proof.

Then using the greedy algorithm in~\cite{calinescu2011maximizing}, we obtain a committee $S_2$ of size at most $(k-|S_1|)$ with
$\score(S_2)\geq (1-1/e) OPT'\geq (1-1/e-o(1))OPT.$
Finally, we arbitrarily select a set $S_3$ of $k-|S_1\cup S_2|$ candidates from $C\setminus (S_1\cup S_2)$.
Since the total score function is monotone, $S_1\cup S_2\cup S_3$ is a $(1-1/e-o(1))$-approximation feasible solution.
\end{proof}
%

\subsection{The Case of $\Delta=1$}
\label{sec:degree1}

In this section, we consider the  case $\Delta=1$, i.e., $P_i\cap P_j=\emptyset$ for all $1\leq i<j\leq p$ and prove the following theorem.

\begin{theorem} (Algorithm for $\Delta=1$)
\label{thm:delta=1}
Given a feasible constrained MS multiwinner voting instance with $\Delta=1$, there exists a randomized polynomial-time $(1-1/e)$-approximation algorithm.
\end{theorem}

\topic{Algorithm DegreeOne}
\begin{enumerate}
\item Compute a fractional solution $y\in B$ by the continuous greedy algorithm in \cite[Section 3.1]{calinescu2011maximizing}.
\item For $1\leq j\leq p$, iteratively do the following until $P_j$ has at most $1$ non-integral coordinate: Arbitrarily select two fractional coordinates  $i,i'$ with $c_i,c_{i'}\in P_j$.
Let $\delta_1=\min\left\{1-y_i,y_{i'}\right\}$ and  $\delta_2=\min\left\{y_i,1-y_{i'}\right\}$. Construct two vectors $y_1=y+\delta_1(\e_i-\e_{i'})$ and $y_2=y+\delta_2 (\e_{i'}-\e_i)$. Let $y\leftarrow y_1$ with probability $\nicefrac{\delta_2}{\delta_1+\delta_2}$ and $y\leftarrow y_2$ otherwise.
\item W.l.o.g., assume $y_1,\ldots,y_{\gamma}$ are the remaining fractional coordinates.
Iteratively do the same procedure as in Step 2 until $y$ becomes an integral solution.
\item Output $S$ whose indicator vector is $y=\mathbf{1}_S$.
\end{enumerate}

\noindent
\begin{proof} [of Theorem~\ref{thm:delta=1}]
We first claim the output $S$ is feasible.
Observe that $\sum_{i\in [m]}y_i=k$ always holds, since Steps 2 and 3 do not change this value.
On the other hand, $a_j=\sum_{i\in P_j}y_i$ is invariant throughout Step 2 and rounds to $\left\lfloor a_j \right\rfloor$ or $\left\lceil a_j \right\rceil$ after Step 3.
  Since $\ell_j\leq a_j\leq u_j$ because $y\in B$ in Step 1, the committee $S$ satisfies all fairness constraints.

For the approximation ratio, we have $F(y)\geq (1-1/e)OPT$ in Step 1 by \cite[Appendix A]{calinescu2011maximizing}.
Then by Lemma~\ref{claim:convex}, we have
\begin{eqnarray*}
& &\Exp[F(y)]=\nicefrac{\delta_2}{\delta_1+\delta_2} \cdot F(y_1)+ \nicefrac{\delta_1}{\delta_1+\delta_2} \cdot F(y_2) \\
&\geq& \left(\nicefrac{\delta_2}{\delta_1+\delta_2}+\nicefrac{\delta_1}{\delta_1+\delta_2}\right)
F\big(y+\nicefrac{\delta_2}{\delta_1+\delta_2}\cdot\delta_1(\e_i-\e_{i'}) +\nicefrac{\delta_1}{\delta_1+\delta_2}\cdot\delta_2(\e_{i'}-\e_{i})\big) =F(y)
\end{eqnarray*}
for each iteration of Step 2 and 3. Thus, we conclude that $\Exp[\score(S)]\geq (1-1/e)OPT$.

For the running time, Step 1 completes in polynomial time by the continuous greedy algorithm proposed in \cite{calinescu2011maximizing}. Step 2 and Step 3 only cost $O(m)$ time in total. Thus, the algorithm runs in polynomial time.
\end{proof}

\noindent
We remark that for the special case of the SNTV rule, there exists an algorithm to compute the optimal solution in $O(m\log m+n)$ time; see Appendix~\ref{app:SNTVdelta=1} for details.

\section{Empirical Results}
\label{sec:experiment}

We compare the performance of winning committees without fairness constraints and with fairness constraints, under several commonly used MS multiwinner voting rules. 
We assume that the voter preferences are generated according to a two-dimensional Euclidean model as in~\cite{schofield2007spatial} which suggests that voters' political
opinions can be described sufficiently well in two dimensions.
Multiwinner voting rules may introduce or exacerbate bias in such preference models. 
We show that the fairness constraints can prevent this, and moreover, the score attained by the fair result is close to the score attained by the optimal unconstrained (and hence, biased) committee.

\subsection{Setup}
\label{sec:setup}

\paragraph{Voting Rules.}
We consider five MS multiwinner voting rules: SNTV, Bloc, $k$-Borda, $\alpha$-CC and $\beta$-CC.
The definitions of SNTV, $\alpha$-CC and $\beta$-CC appear in Section~\ref{sec:pre}.
{Bloc} outputs the $k$ candidates that have the most voters listing them in the top $k$ of their list.
{$k$-Borda} outputs the $k$ candidates with the largest sum of the Borda scores (Definition~\ref{def:betaCC}) that she receives from all voters.
As a baseline, we also consider the Random voting rule that simply selects a committee uniformly at random.
In all cases, we let $k=12$ be the size of the desired committee.

\paragraph{Sampling Candidates, Voters and Preferences.}
We generate 400 voters and 120 candidates where each voter and candidate is represented by a point in the $[-3,3]\times [-3,3]$ square.
1/4 of the voters are sampled uniformly at random from each quadrant, and 
1/3 of the candidates are sampled uniformly from the first quadrant, 1/4 the second quadrant, 1/6 the third quadrant, and 1/4 the fourth quadrant.
As in \cite{elkind2017multiwinner}, we use Euclidean preferences; given a pair of candidates $c_i,c_j\in \R^2$, a voter $a\in \R^2$ prefers $c_i$ to $c_j$ if $d\left(c_i,a\right)<d\left(c_j,a\right)$ where $d\left(\cdot,\cdot\right)$ is the Euclidean distance.

\paragraph{Groups and Fairness Constraints.}
We consider each quadrant to be a different group, and let $P_i$ be the collection of candidates in the i-th quadrant. We compare the performance of three different types of constraints to the unconstrained optimal solution and the random baseline.
\begin{itemize}
	\item \textbf{Proportional w.r.t. voters:} the number of winners is proportional to the number of voters in each quadrant. 
	\item \textbf{Proportional w.r.t. candidates:} The number of winners is proportional to the number of candidates in each quadrant. 
	\item \textbf{Relax:} The number of winners in each quadrant is allowed to be anywhere in the range between what would be proportional to voters and proportional to candidates. 
\end{itemize}

\paragraph{Metrics.}
As a metic of fairness, we consider the Gini index, a well-studied metric for inequality: 
Let $p$ be the total number of groups, and let $n_i$ be the number of winners in each group $P_i$, then the Gini index is defined to be
$\frac{\sum_{i=1}^p \sum_{j=1}^p \left|n_i-n_j \right|}{2p\sum_{i=1}^p n_j}$.
This measures, on a scale from 0 to 1, how disproportionate the distribution of winners is amongst the groups (here, quadrants), with 1 meaning complete inequality and 0 meaning complete equality.
In this context, the Gini index should be interpreted relationally, observing how different methods and constraints change the index up or down, as opposed to thinking of it as prescribing a correct outcome.

In addition, for each voting method we measure the price of fairness, i.e., the ratio between its score and the optimal unconstrained score. A ratio close to 100\% indicates that despite adding fairness constraints the method attains close to an optimal score.

\subsection{Results}
\label{sec:analysis}

We report the mean and the standard deviation of the Gini index, and the mean of the score ratio of 1000 repetitions in Table \ref{tab:experiment}, and depict the outcome of 1000 repetitions in Table~\ref{tab:visualization}.

\paragraph{Unconstrained rules introduce bias.}

We observe that there can be bias in the unconstrained voting rules as the Gini index is large; in particular, for several rules, the inequality is larger than that of a random set of candidates suggesting that underlying disproportionalities in the set of candidates can be exacerbated by certain rules.
For instance, observe the Bloc rule in Table~\ref{tab:visualization}. The unconstrained optimal committee only includes few winners in the first quadrant.
The intuition is that having many candidates in the same quadrant thins their supporters due to competition, leading to disproportionately fewer representatives as winners.

For $\alpha$-CC and $\beta$-CC, the unconstrained Gini index is close to that the Gini index when constraints are placed to be proportional with respect to candidates. 
However, the standard deviation in both cases is 0.06, which suggests that there can be significant bias in some outcomes of this process. 
The constrained process, on the other hand, will always satisfy the corresponding Gini index exactly.

The fairness constraints allow us to fix a particular distribution of winners (and hence corresponding Gini index), in this case either proportional with respect to voters or candidates, and hence allows us to de-bias the result in any desired manner. The result of the relaxed constraints, which allow the distribution of winners within these ranges 
vary with some rules tending towards proportionality with respect to voters and others with respect to candidates.

\paragraph{The price of fairness is small.}
Because the feasible space of committees becomes smaller in the constrained settings, the optimal constrained score may be less than its unconstrained counterpart.
However, in Table \ref{tab:experiment} we observe that the fairness constraints do not decrease the score by much; in fact 
for the $\alpha$-CC or $\beta$-CC rules the score does not decrease at all.
\footnote{The ratio $100\%$ for $\beta$-CC is due to rounding.} 
This is not just a matter of there being many good committees; indeed a random committee in these settings does not perform well in comparison.
Hence, the ``price of fairness'' is small in this setting; understanding this quantity more generally remains an important direction for future work.

\begin{table}
	\centering
	
	\begin{tabular}{p{0.01\textwidth} p{0.1\textwidth} p{0.1\textwidth} p{0.1\textwidth} p{0.1\textwidth} p{0.1\textwidth} p{0.1\textwidth} p{0.1\textwidth} p{0.1\textwidth}}
		& \multicolumn{2}{c}{Unconstrained} & \multicolumn{2}{c}{Prop w.r.t. voters} & \multicolumn{2}{c}{Prop w.r.t. candidates} & \multicolumn{2}{c}{Relax} \\
		\rotatebox{90}{SNTV}  & \fbox{\includegraphics[keepaspectratio, width=\linewidth]{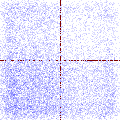}} &  \fbox{\includegraphics[keepaspectratio, width=\linewidth]{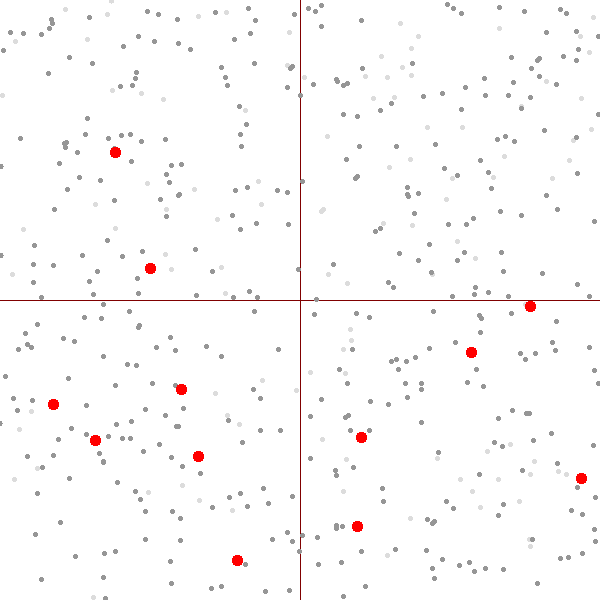}} & \fbox{\includegraphics[keepaspectratio, width=\linewidth]{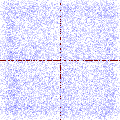}} & \fbox{\includegraphics[keepaspectratio, width=\linewidth]{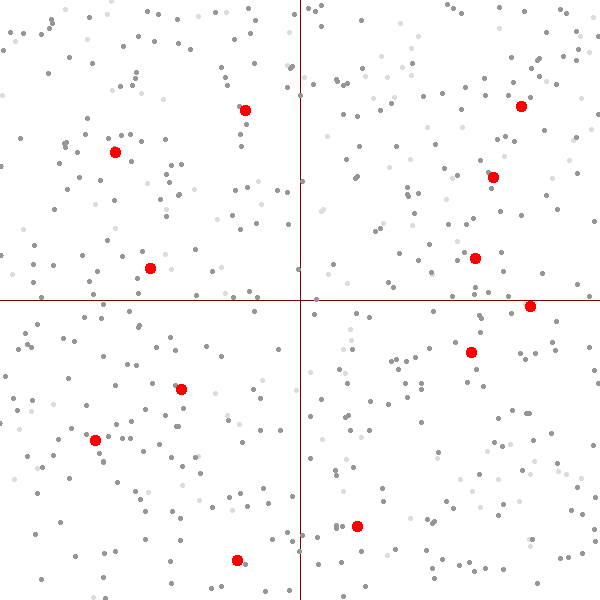}} & \fbox{\includegraphics[keepaspectratio, width=\linewidth]{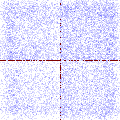}} & \fbox{\includegraphics[keepaspectratio, width=\linewidth]{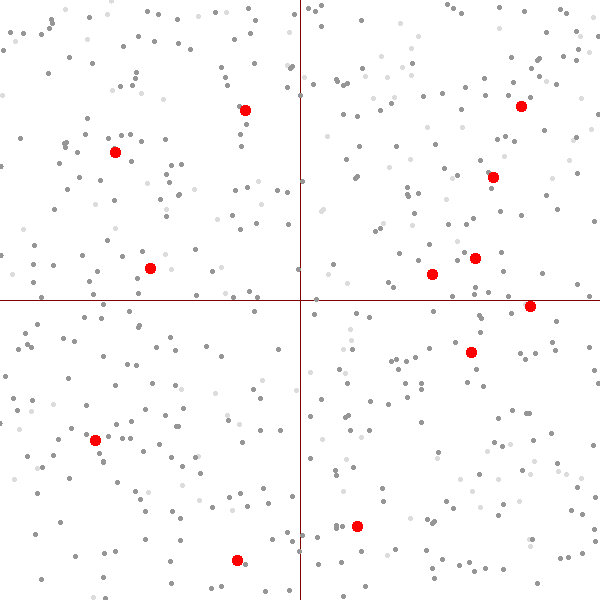}} & \fbox{\includegraphics[keepaspectratio, width=\linewidth]{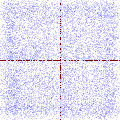}} & \fbox{\includegraphics[keepaspectratio, width=\linewidth]{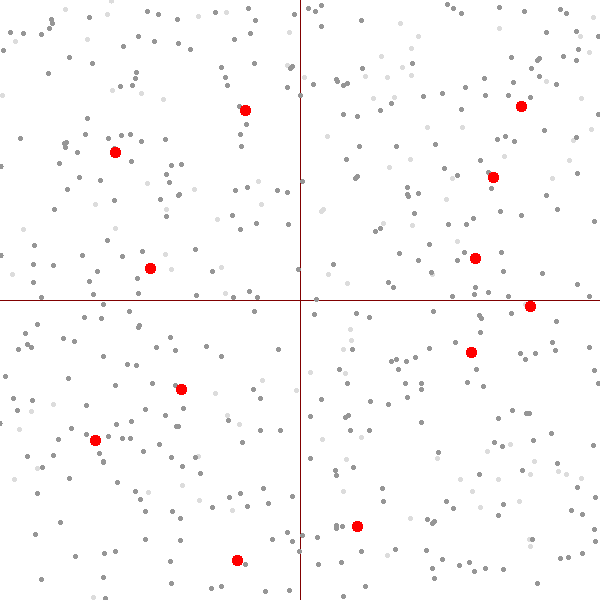} } \\
		\rotatebox{90}{Bloc}  & \fbox{\includegraphics[keepaspectratio, width=\linewidth]{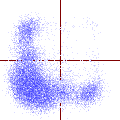}} &  \fbox{\includegraphics[keepaspectratio, width=\linewidth]{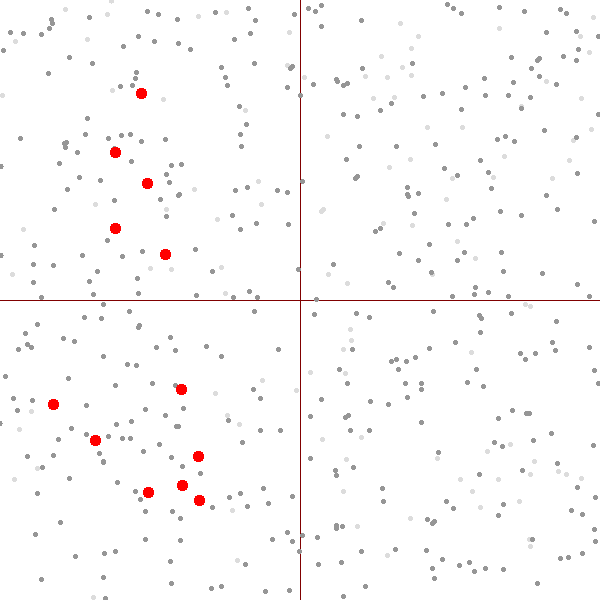}} & \fbox{\includegraphics[keepaspectratio, width=\linewidth]{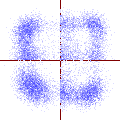}} & \fbox{\includegraphics[keepaspectratio, width=\linewidth]{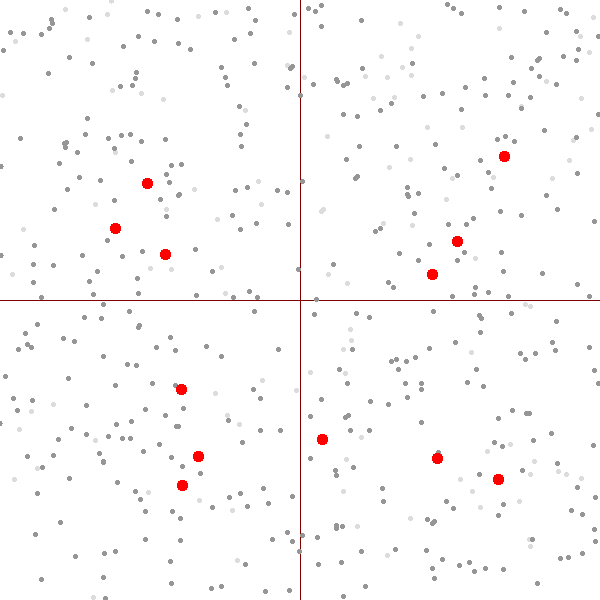}} & \fbox{\includegraphics[keepaspectratio, width=\linewidth]{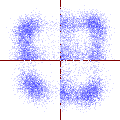}} & \fbox{\includegraphics[keepaspectratio, width=\linewidth]{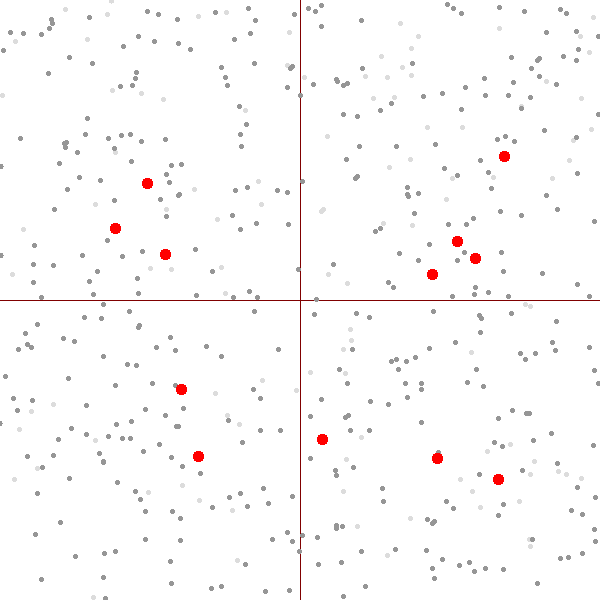}} & \fbox{\includegraphics[keepaspectratio, width=\linewidth]{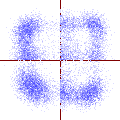}} & \fbox{\includegraphics[keepaspectratio, width=\linewidth]{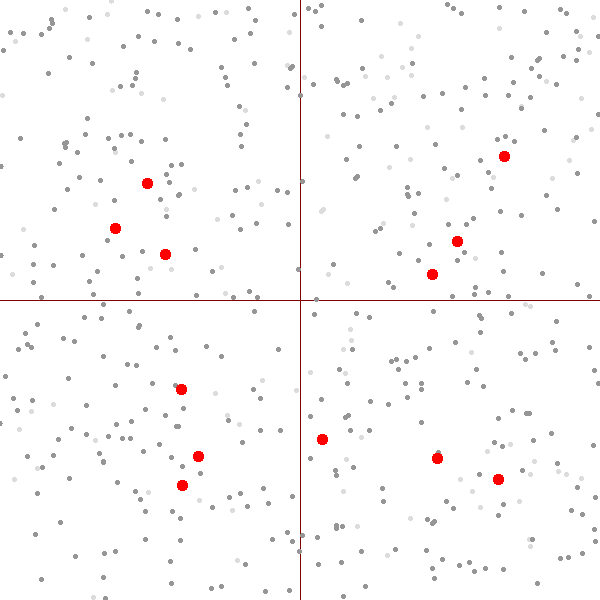} } \\
		\rotatebox{90}{$k$-Borda}  & \fbox{\includegraphics[keepaspectratio, width=\linewidth]{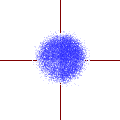}} &  \fbox{\includegraphics[keepaspectratio, width=\linewidth]{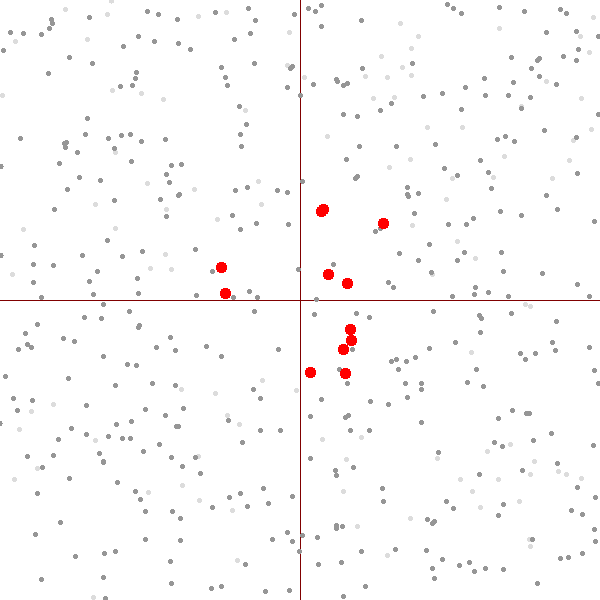}} & \fbox{\includegraphics[keepaspectratio, width=\linewidth]{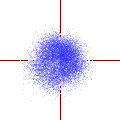}} & \fbox{\includegraphics[keepaspectratio, width=\linewidth]{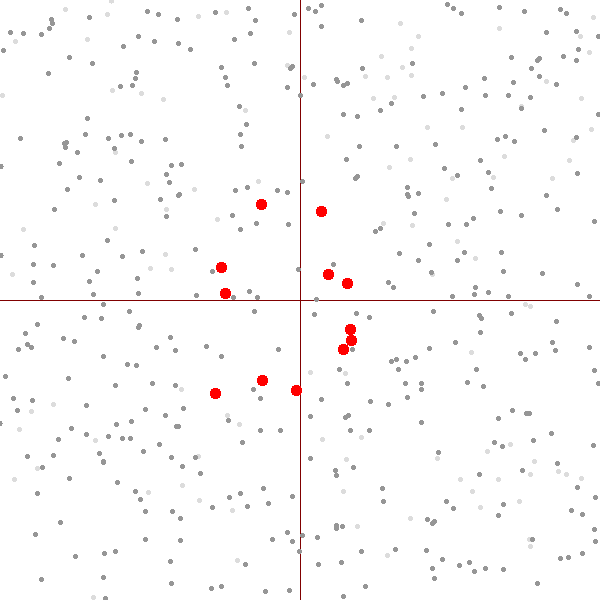}} & \fbox{\includegraphics[keepaspectratio, width=\linewidth]{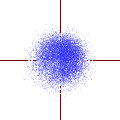}} & \fbox{\includegraphics[keepaspectratio, width=\linewidth]{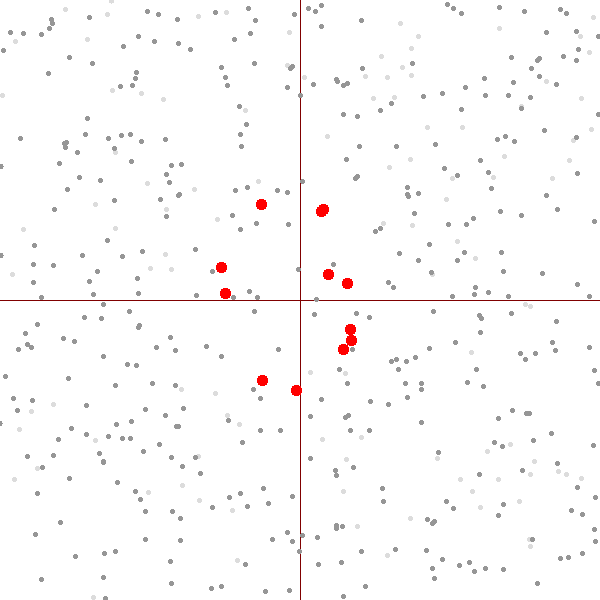}} & \fbox{\includegraphics[keepaspectratio, width=\linewidth]{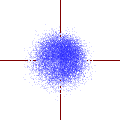}} & \fbox{\includegraphics[keepaspectratio, width=\linewidth]{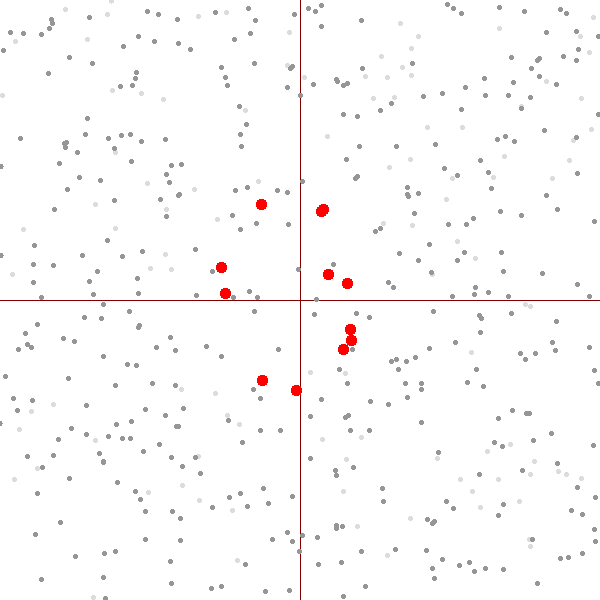} } \\
		\rotatebox{90}{$\alpha$-CC}  & \fbox{\includegraphics[keepaspectratio, width=\linewidth]{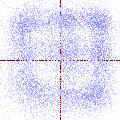}} &  \fbox{\includegraphics[keepaspectratio, width=\linewidth]{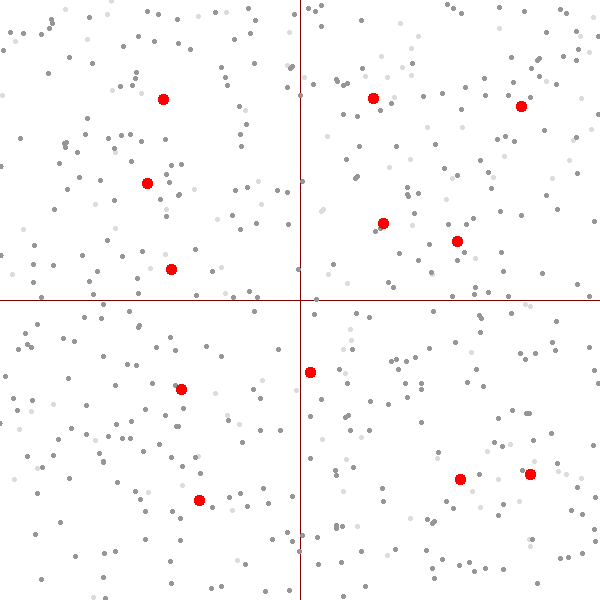}} & \fbox{\includegraphics[keepaspectratio, width=\linewidth]{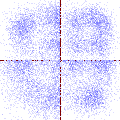}} & \fbox{\includegraphics[keepaspectratio, width=\linewidth]{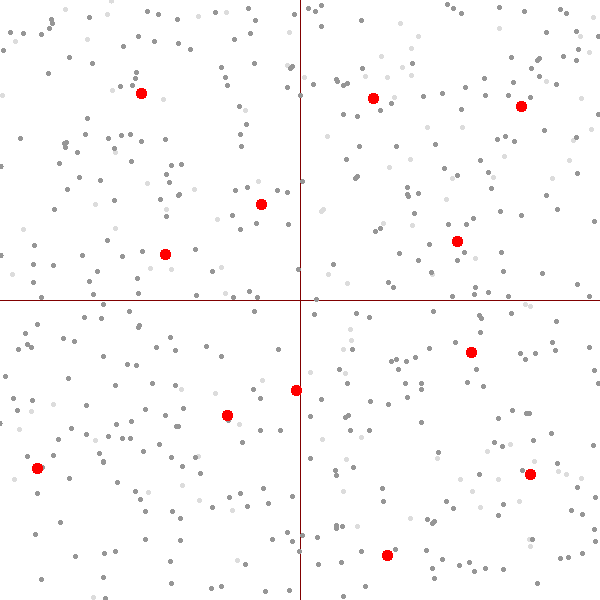}} & \fbox{\includegraphics[keepaspectratio, width=\linewidth]{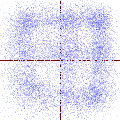}} & \fbox{\includegraphics[keepaspectratio, width=\linewidth]{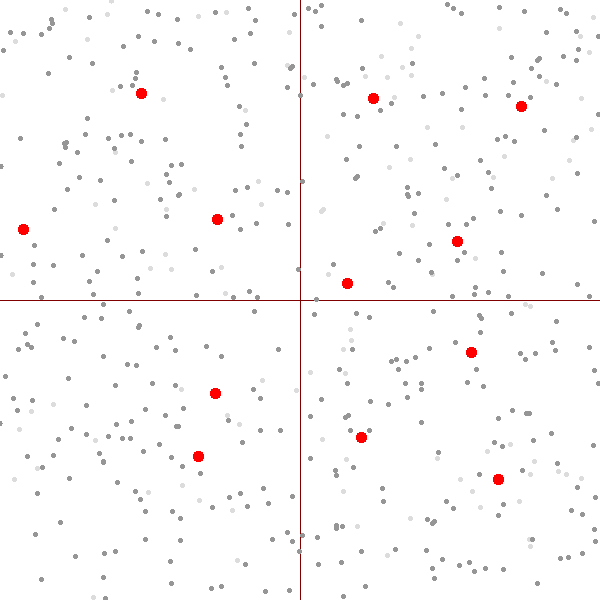}} & \fbox{\includegraphics[keepaspectratio, width=\linewidth]{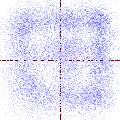}} & \fbox{\includegraphics[keepaspectratio, width=\linewidth]{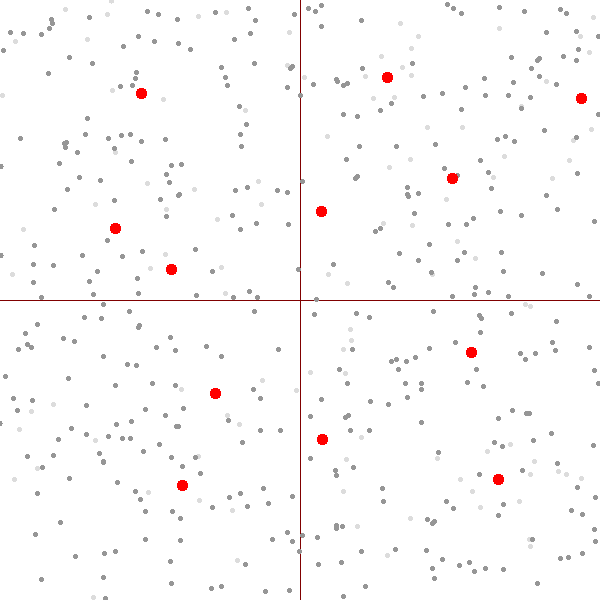} } \\
		\rotatebox{90}{$\beta$-CC}  & \fbox{\includegraphics[keepaspectratio, width=\linewidth]{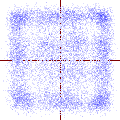}} &  \fbox{\includegraphics[keepaspectratio, width=\linewidth]{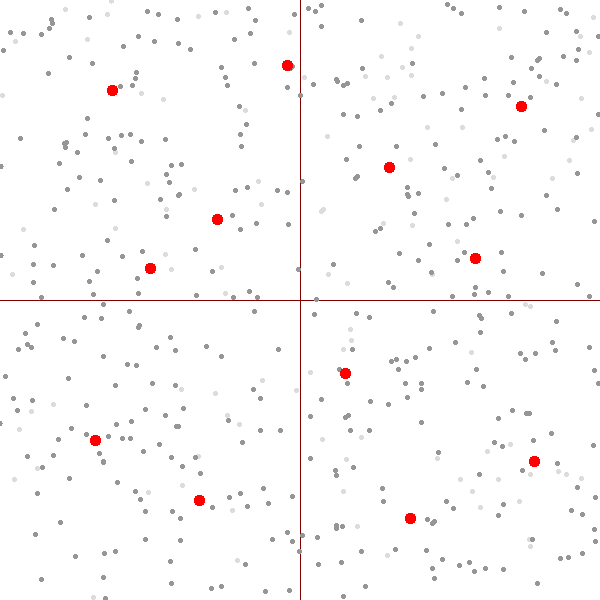}} & \fbox{\includegraphics[keepaspectratio, width=\linewidth]{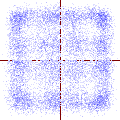}} & \fbox{\includegraphics[keepaspectratio, width=\linewidth]{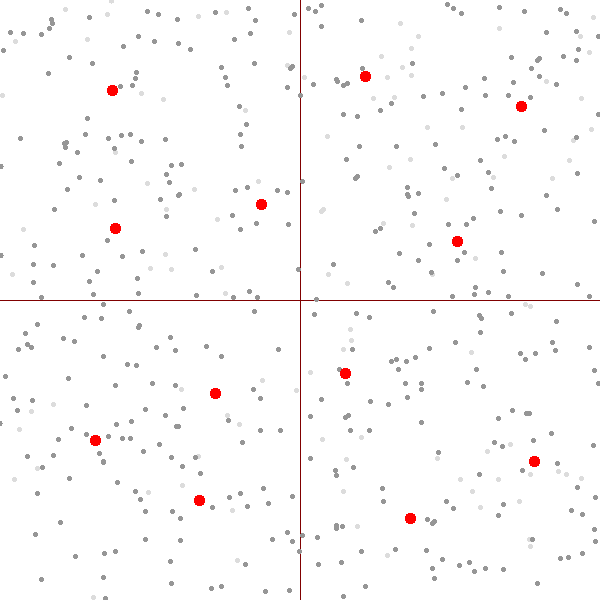}} & \fbox{\includegraphics[keepaspectratio, width=\linewidth]{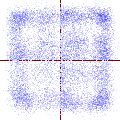}} & \fbox{\includegraphics[keepaspectratio, width=\linewidth]{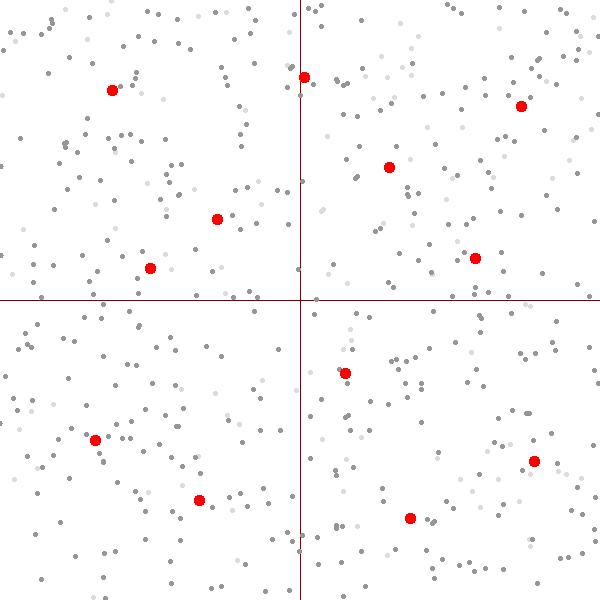}} & \fbox{\includegraphics[keepaspectratio, width=\linewidth]{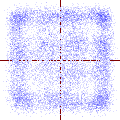}} & \fbox{\includegraphics[keepaspectratio, width=\linewidth]{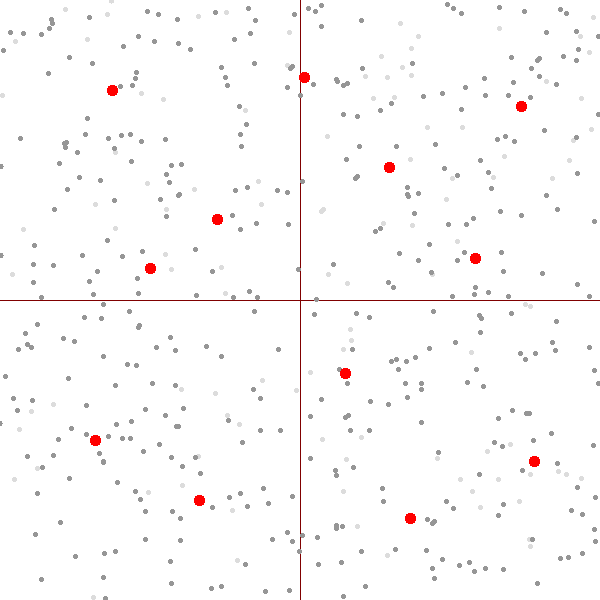} }
	\end{tabular}
	
	\caption{
		Depicts the winners from 1000 experiments for different rules and different fairness constraints.
		For sample election, dark gray dots represent voters, light gray dots represent candidates, and red dots represent the winners.}
	\label{tab:visualization}
\end{table}

\section{Conclusion}
\label{sec:conclusion}

In this paper, we propose a general model for multiwinner
voting that can handle arbitrary MS score functions, can
simultaneously ensure fairness across multiple sensitive attributes
and can satisfy many different existing notions of fairness.
We also develop efficient approximation algorithms for
our model in different settings. The empirical results show
that our model indeed eliminates bias compared to unconstrained
voting, with a controlled price of fairness.

\begin{table*}[t] 
	\footnotesize{
		\begin{tabular}{|*{11}{c|}}
			\multirow{2}{*}{Rule} &  \multicolumn{2}{c|}{Unconstrained} & \multicolumn{2}{c|}{Prop w.r.t. voters} & \multicolumn{2}{c|}{Prop w.r.t. candidates} & \multicolumn{2}{c|}{Relax} & \multicolumn{2}{c|}{Random} \\ 
			& Gini & \% Opt & Gini & \% Opt & Gini & \% Opt & Gini & \% Opt & Gini & \% Opt \\ 
			\hline
			SNTV & 0.24 (0.09)& 100  & 0 (0) & 97.0  & 0.125 (0)& 94.2  & 0.01 (0.02) & 97.0  & 0.22 (0.09)& 37.1  \\  \cline{1-11}
			Bloc & 0.28 (0.10)& 100   & 0 (0)& 91.6  & 0.125 (0)& 88.4  & 0.00 (0.00)& 91.6  & 0.22 (0.09) &  61.9  \\ \cline{1-11}
			$k$-Borda & 0.24 (0.09)& 100   & 0 (0) & 98.9  & 0.125 (0)& 99.3  & 0.11 (0.04)& 99.3  & 0.22 (0.09) & 72.6  \\  \cline{1-11}
			$\alpha$-CC & 0.15 (0.06)& 100   & 0 (0)& 100  & 0.125 (0)& 100  & 0.10 (0.05)& 100  & 0.22 (0.09)& 73.5  \\  \cline{1-11}
			$\beta$-CC & 0.11 (0.06)& 100   & 0 (0)& 100  & 0.125 (0)& 100  & 0.07 (0.06)& 100  & 0.22 (0.09)& 95.8  \\   \hline
		\end{tabular}
	}
	\label{tab:experiment}	
	\caption{Under different rules, the average Gini index, the standard deviation of Gini index, and the average ratio of the constrained optimal score over the unconstrained optimal score. 
	}
\end{table*}

\section*{Acknowledgment}
We would like to thank Jian Li for discussions on submodular optimization.

\bibliography{fairvoting}
\bibliographystyle{plain}

\appendix

\section{APX-Hardness for $\alpha$-CC and $\beta$-CC: $\Delta=1$ and $p=2$}
\label{app:APXhard}

In this section, we consider a special case of $\Delta=1$ and $p=2$. We prove that the constrained $\beta$-CC multiwinner voting problem is $(1-1/e)$-inapproximable. This is in contrast with the unconstrained $\beta$-CC multiwinner voting problem which has a PTAS \cite{skowron2015achieving}.

\begin{theorem} (APX-hardness for $\alpha$-CC and $\beta$-CC)
\label{thm:simplelower}
Given the constrained $\alpha$-CC or $\beta$-CC multiwinner voting problem with $\Delta=1$ and $p=2$,  there is no polynomial-time $(1-1/e+\varepsilon)$-approximation algorithm for any constant $\eps>0$ unless $P=NP$.
\end{theorem}
\noindent
\begin{proof}
The unconstrained $\alpha$-CC multiwinner voting problem is equivalent to the maximum $k$-coverage problem, which is $(1-1/e)$-inapproximable by \cite{feige1998threshold}. Hence we focus on the $\beta$-CC rule. If $m\leq 3/\varepsilon$, we can enumerate all possible committees (at most $2^m=O(1)$ committees) and compute the optimal solution. Therefore, we only need to consider the case of $m\geq 3/\varepsilon$. We present a reduction from the maximum coverage problem.

Let $S_1,\ldots,S_m\subseteq [n]$ be $m$ sets and let $k$ be an integer for which we want to pick $k$ sets to cover as many elements as possible. We construct an instance of the constrained $\beta$-CC multiwinner voting problem as follows.
Construct a corresponding candidate $c_i$ for each set $S_i$, and construct $nm^2-m$ additional candidates $c_{m+1},\ldots,c_{nm^2}$. For each element $j$, construct a corresponding voter $a_i$. So there are  $nm^2$ candidates and $n$ voters in total. We define the preference order of each voter $a_i$ as follows.
\begin{enumerate}
\item W.l.o.g., assume $i$ belongs to $S_1,\ldots,S_t$. Let the favourite $t$ candidates of $a_i$ be $c_1,\ldots,c_t$ (the order can be arbitrary).
\item Let the next $nm^2-m$ candidates in the preference order of $a_i$ be $c_{m+1},\ldots,c_{nm^2-m}$ (the order can be arbitrary).
\item Let the remaining $m-t$ candidates in the preference order of $a_i$ be $c_{t+1},\ldots,c_m$ (the order can be arbitrary).
\end{enumerate}
\noindent
Finally, define $P_1=\left\{c_1,\ldots,c_m\right\}$ with $l_1=u_1=k$, and $P_2=\left\{c_{m+1},\ldots,c_{nm^2}\right\}$ with $l_2=u_2=0$.

In the following, we assume there exists a $(1-1/e+\varepsilon)$-approximation algorithm for the above constrained $\beta$-CC multiwinner voting instance. For the constrained $\beta$-CC multiwinner voting instance, suppose we output a committee $SOL$ with $\score(SOL)\geq (1-1/e+\varepsilon) OPT$.
For the  maximum $k$-coverage instance, we claim that $SOL'=\left\{S_j\mid c_j\in SOL\right\}$ is a $(1-1/e)$-approximation solution. Let $O'$ denote the optimal solution of the maximum $k$-coverage instance, and $\eta$ denote the number of elements covered by $O'$. Consider the committee $\left\{c_j\mid S_j\in OPT'\right\}$. The total score of this committee is at least $\eta(nm^2-m)$, which implies
\begin{eqnarray*}
\score(SOL)&\geq& (1-1/e+\varepsilon) OPT \\
&\geq& (1-1/e+\varepsilon)\eta(nm^2-m).
\end{eqnarray*}
If $SOL'$ covers less than $(1-1/e)\eta$ elements, we have
$$
\score(SOL)< (1-1/e)\eta (nm^2-1)+nm.
$$
Since $m\geq 3/\varepsilon$, the above two inequalities can not hold at the same time. Therefore, $SOL'$ must cover at least $(1-1/e)\eta$ elements, which implies that there exists a $(1-1/e)$-approximation algorithm for the maximum $k$-coverage problem. This completes the proof.
\end{proof}
%

\section{The Case that $p$ is Constant}
\label{app:constant}

Sometimes, the number of groups under consideration is small. In this case, we have the following theorem.

\begin{theorem} (approximation algorithm: $p$ is constant)
\label{thm:constant}
Assume $p$ is constant. There exists a $(1-1/e)$-approximation algorithm for the constrained multiwinner voting problem. Especially, there exists a polynomial-time algorithm to compute the optimal solution if using the SNTV rule.
\end{theorem}
\noindent
\begin{proof} [Proof of Theorem~\ref{thm:constant}]
Let $T_i:=\left\{j\in [p]: i\in P_j\right\}$ be the set of types that candidate $c_i$ has. Let $q$ be the number of distinct $T_i$s. Since $p$ is constant, $q\leq 2^p$ is also constant. For each $T_i$, construct a corresponding set $V_i\subseteq C$ of candidates whose types are $T_i$. All $V_i$ form a partition of $C$.
We call an integral vector $\mu=(k_1,\ldots,k_q)$ an ``available vector" if it is feasible to select exactly $k_i$ candidates from $V_i$ for all $i\in [q]$. Note that there are at most $m^q$ available vectors.

For each available vector $\mu=(k_1,\ldots,k_q)$, we construct a constrained multiwinner voting instance: Denote $V_1,\ldots,V_q$ to be the groups and $\calB_{\mu}$ to be the collection of committees satisfying $|S\cap V_i|\leq k_i$ $\forall i\in [q]$. The objective is to find a committee $S\in \calB_{\mu}$ which maximizes $\score(S)$. Let $OPT_{\mu}$ denote the optimal score of this problem. We find this problem is equivalent to maximizing a monotone submodular function constrained to a partition matroid. By \cite{calinescu2011maximizing}, there exists a polynomial-time $(1-1/e)$-approximation algorithm.
Hence it is able to compute a committee $S_{\mu}$ of size at most $k$ with $\score(S_\mu)\geq (1-1/e)OPT_{\mu}$. Then we expand $S_{\mu}$ until there are exactly $k_i$ candidates from each $V_i$. By monotonicity, this procedure can not decrease $\score(S_\mu)$. Since $\mu=(k_1,\ldots,k_q)$ is an available vector, the committee $S_{\mu}$ must be feasible.
Among all committees $S_{\mu}$, we select the one with the largest total score. Since the optimal committee also corresponds to some available vector, our selected committee must be a $(1-1/e)$-approximation solution.

Especially, the total score function for the SNTV rule is linear. Therefore, we can greedily compute the optimal feasible committee corresponding to each available vector. Among all such committees, we select the one with the largest total score, which is exactly the optimal solution.
\end{proof}

\section{An Algorithm for Constrained SNTV Multiwinner Voting when $\Delta=2$}
\label{app:SNTV}

In this section, we consider the case of $\Delta=2$ using the SNTV rule and give a polynomial-time algorithm.

\begin{theorem} (exact algorithm for SNTV: $\Delta=2$)
\label{thm:SNTV2}
If $\Delta=2$, there exists a polynomial-time algorithm for the constrained SNTV multiwinner voting problem.
\end{theorem}

\noindent
\begin{proof}
We reduce to the weighted perfect $b$-matching problem, which has a polynomial-time algorithm \cite{schrijver2002combinatorial}. For each candidate $a$, let $w_a$ denote the total number of voters most preferred to it.
W.l.o.g., we assume that each candidate has exactly two attributes. Otherwise, if a candidate $c_i$ has only one attribute, we construct a new group $\left\{c_i\right\}$ and set its constrained lower bound to be 0 and its upper bound to be 1. Similarly, we construct two groups for those candidates without types.
In the following, we construct a graph $G=(V^1\cup V^2\cup V^3,E)$ and a weighted perfect $b$-matching instance.
\begin{enumerate}
\item For each group $P_i$, we construct a corresponding vertex $v_i$. Let $V^1=\left\{v_1,\ldots,v_p\right\}$ be the collection of these vertices. For each candidate $a$, suppose it belongs to $P_i$ and $P_j$. We construct a corresponding edge $e_a=(v_i,v_j)$ on $G$.
\item For each vertex $v_i\in V^1$, construct $u_i-\ell_i$ vertices adjacent to $v_i$. Let $V^2$ be the collection of all these vertices.
\item Construct a collection $V^3$ of $2k-\sum_{i\in [p]}\ell_i$ vertices. For each pair $u\in V^2,v\in V^3$, construct an edge.
\item For each vertex $v_i\in V^1$, set $b(v_i)=u_i$. For each vertice $v\in V^2\cup V^3$, set $b(v)=1$.
\item For each edge $e_a\in V^1\times V^1$, set its weight $w(e_a)=w_a$. Set the weights of other edges to be 0.
\end{enumerate}
\noindent
We claim that any committee $S\in \calB$ corresponds to a perfect $b$-matching of weight equal to $\score(S)$. Given a committee $S$, we do the following steps.
\begin{enumerate}
\item Pick all edges $e_a$ if $a\in S$.
\item For each vertex $v_i\in V^1$, arbitrarily pick $u_i-|S\cap P_i|$ edges in $v_i \times V^2$. Then $v_i$ is covered by exactly $b(v_i)=u_i$ times. Moreover, these edges cover $\sum_{i\in [p]} u_i-|S\cap P_i|=\sum_{i\in [p]} u_i-2k$ vertices in $V^2$ in total.
\item There are $2k-\sum_{i\in [p]} \ell_i$ uncovered vertices in $V^2$. Let $V'$ denote the collection of these vertices. We pick $2k-\sum_{i\in [p]} \ell_i$ edges to form a perfect matching among $V'\cup V^3$.
\end{enumerate}
\noindent
Therefore, each vertex in $V^2\cup V^3$ is covered once which implies the above construction is a perfect $b$-matching. Moreover, the total weight is exactly equal to $\score(S)$.

On the other hand, we show that any perfect $b$-matching $M$ of weight $w(M)$ corresponds to a committee of size $k$ with score $w(M)$. We claim that $M$ must contain $k$ edges in $V^1\cup V^1$. If this claim is true, consider the following committee $S=\left\{a:e_a\in M\right\}$ of score $w(M)$. Since there are only $u_i-\ell_i$ edges adjacent to $v_i$ from $V^2$ $(i\in [p])$, $M$ must contain $\ell_i\leq t\leq u_i$ edges in $v_i \times V^1$. Therefore, we have $\ell_i\leq |S\cap P_i|\leq u_i$ for all $i\in [p]$, which completes the proof that $S$ is a feasible committee.

It remains to prove the claim. Observe that $M$ must contain $2k-\sum_{i\in [p]} \ell_i$ edges in $V^2\times V^3$ since each vertex in $V^3$ should be covered once. Moreover, due to the fact that $|V^2|=\sum_{i\in [p]} u_i-\ell_i$, $M$ must contain $\sum_{i\in [p]} u_i-2k$ edges in $V^1\times V^2$ for covering each vertex in $V^2$ once. These edges cover vertices in $V^1$ by $\sum_{i\in [p]} u_i-2k$ total times. Since $M$ should cover vertices in $V^1$ by $\sum_{i\in [p]} u_i$ times, there must be $k$ edges in $V^1\times V^1$. It completes the proof.
\end{proof}

\noindent
From the above proof, we have the following corollary.

\begin{corollary}
The constrained multiwinner voting feasibility problem can be solved in polynomial time when $\Delta=2$.
\end{corollary}

\section{A Nearly Linear Algorithm for Constrained SNTV Multiwinner Voting when $\Delta=1$}
\label{app:SNTVdelta=1}

In this section, we consider a simple case of $\Delta=1$ using the SNTV rule, and present a nearly linear exact algorithm.

\begin{theorem}
\label{thm:SNTVdelta=1}
There exists an $O(n+m\log m)$ time algorithm for the constrained SNTV multiwinner voting problem when $\Delta=1$.
\end{theorem}

\begin{proof}
Note that the SNTV rule has a linear total score function.
For each candidate $c_i$, we first compute its weight $w_i$ which is the number of voters most preferred to $c_i$.
This step costs $O(n)$ time.
Among each group $P_i$, we select $\ell_i$ candidates, whose weights $w_i$ are largest, into the committee $S$ in $O(m)$ time.
Finally, we sort the remaining candidates in nondecreasing order by their weights $w_i$ (breaking ties arbitrarily) with running time $O(m\log m)$.
Consider candidates in order and each time put a candidate into $S$ if keeping feasibility, i.e., $|S\cap P_i|\leq u_i$ always holds for all $i\in [p]$. Since $\Delta=1$, it costs $O(1)$ time to check the feasibility for each candidate.
Therefore, the total time is $O(n+m\log m)$.
The optimality follows from the linear objective function, which completes the proof.
\end{proof}

\end{document}